%% file: upload-large-games-arxiv.tex
\newif\ifdraft \drafttrue
\newif\iffull \fulltrue

\makeatletter \@input{tex.flags} \makeatother

\documentclass[11pt]{article}

\usepackage{amssymb}
\usepackage{graphicx}

\usepackage{tikz}  
\usetikzlibrary{arrows}

\usepackage{url}
\urldef{\mailsa}\path|{paul.goldberg, francisco.marmolejocossio}@cs.ox.ac.uk|
\urldef{\mailsb}\path|wuzhiwei@cis.upenn.edu|

\usepackage{algorithm}
\usepackage[noend]{algorithmic}

\newcommand{\algorithmicinit}{\textbf{Initialisation:}}
\newcommand{\INITIALIZATION}{\item[\algorithmicinit]}

\usepackage{fullpage}

\usepackage{amsmath, amssymb, amsthm}
\usepackage{mathtools}
\usepackage{verbatim}
\usepackage{color}
\usepackage{xcolor}

\usepackage{multicol}

\definecolor{DarkGreen}{rgb}{0.1,0.5,0.1}
\definecolor{DarkRed}{rgb}{0.5,0.1,0.1}
\definecolor{DarkBlue}{rgb}{0.1,0.1,0.5}
\usepackage[]{hyperref}
\hypersetup{
    unicode=false,          
    pdftoolbar=true,        
    pdfmenubar=true,        
    pdffitwindow=false,      
    pdftitle={},    
    pdfauthor={}
    pdfsubject={},   
    pdfnewwindow=true,      
    pdfkeywords={keywords}, 
    colorlinks=true,       
    linkcolor=DarkRed,          
    citecolor=DarkGreen,        
    filecolor=DarkRed,      
    urlcolor=DarkBlue,          
}
\usepackage{xspace}
\usepackage{cleveref}
\usepackage{thmtools, thm-restate}
\usepackage[numbers,sort&compress]{natbib}

\newcommand\cA{\mathcal{A}}

\newcommand\cP{\mathcal{P}}
\newcommand\cQ{\mathcal{Q}}

\DeclareMathOperator*{\Expectation}{\mathbb{E}}
\newcommand{\Ex}[2]{\Expectation_{#1}\left[#2\right]}

\newcommand{\eps}{\varepsilon}
\def\epsilon{\varepsilon}

\renewcommand{\hat}{\widehat}

\DeclareMathOperator*{\argmax}{\mathrm{argmax}}

\newtheorem*{theorem*}{Theorem}
\newtheorem*{observation*}{Observation}
\declaretheorem[
  name=Theorem,
  refname={theorem, theorems},
  Refname={Theorem, Theorems}]{theorem}
\declaretheorem[
  name=Lemma,
  refname={lemma, lemmas},
  Refname={Lemma, Lemmas}]{lemma}

\declaretheorem[
  name=Corollary,
  refname={corollary, corollaries},
  Refname={Corollary, Corollaries}]{corollary}
\declaretheorem[
  name=Definition,
  refname={definition, definitions},
  Refname={Definition, Definitions}]{definition}
\declaretheorem[
  name=Observation,
  refname={observation, observations},
  Refname={Observation, Observations}]{observation}

\title{Logarithmic Query Complexity for Approximate Nash Computation in Large Games}

\author{Paul W. Goldberg \thanks{University of Oxford. Emails: \href{mailto:paul.goldberg@cs.ox.ac.uk}{paul.goldberg@cs.ox.ac.uk}, \href{mailto:francisco.marmolejo@cs.ox.ac.uk}{francisco.marmolejo@cs.ox.ac.uk}} \and 
Francisco J. Marmolejo-Coss\'{i}o \footnotemark[1] \thanks{Supported by the Mexican National Council of Science and Technology (CONACyT)} \and
 Zhiwei Steven Wu\thanks{University of Pennsylvania. Email:
    \href{mailto:wuzhiwei@cis.upenn.edu}{wuzhiwei@cis.upenn.edu}.
    }}

\begin{document}

\maketitle

\begin{abstract}
We investigate the problem of equilibrium computation for ``large'' $n$-player games.
Large games have a Lipschitz-type property that no single player's utility
is greatly affected by any other individual player's actions.
In this paper, we mostly focus on the case where any change of strategy by a player
causes other players' payoffs to change by at most $\frac{1}{n}$.
We study algorithms having query access to the game's payoff function,
aiming to find $\epsilon$-Nash equilibria.
We seek algorithms that obtain $\epsilon$ as small as possible, in time polynomial in $n$.

Our main result is a randomised algorithm that achieves $\epsilon$ approaching
$\frac{1}{8}$ for 2-strategy games in a {\em completely uncoupled} setting, where each player
observes her own payoff to a query, and adjusts her behaviour independently
of other players' payoffs/actions. $O(\log n)$ rounds/queries are required.
We also show how to obtain a slight improvement over $\frac{1}{8}$,
by introducing a small amount of communication between the players.

Finally, we give extension of our results to large games with more
than two strategies per player, and alternative largeness parameters.
\end{abstract}

\input{upload-section-intro-arxiv}

\input{upload-section-largegames-arxiv}

\input{upload-section-warmup-arxiv}

\input{upload-section-continuous-dynamics-arxiv}

\input{upload-section-coupled-arxiv}

\input{upload-section-extensions-arxiv}

\input{upload-section-conclusion-arxiv}

\bibliographystyle{plainnat}

\bibliography{./upload-large-games-arxiv.bbl}


\end{document}

%% file: upload-section-intro-arxiv.tex
\section{Introduction}\label{sec:intro}

In studying the computation of solutions of multi-player games,
we encounter the well-known problem that a game's payoff
function has description length exponential in the number of players.
One approach is to assume that the game comes from a
concisely-represented class (for example, graphical games,
anonymous games, or congestion games), and another one is to consider
algorithms that have query access to the game's payoff function.

In this paper, we study the computation of approximate Nash equilibria
of multi-player games having the feature that if a player
changes her behaviour, she only has a small effect on the payoffs
that result to any other player. These games, sometimes called
{\em large} games, or {\em Lipschitz} games, have recently been
studied in the literature, since they model various real-world
economic interactions; for example, an individual's choice of what
items to buy may have a small effect on prices, where other individuals
are not strongly affected. Note that these games do not
have concisely-represented payoff functions, which makes them a
natural class of games to consider from the query-complexity perspective.
It is already known how to compute approximate {\em correlated equilibria}
for unrestricted $n$-player games. Here we study the more demanding solution
concept of approximate Nash equilibrium.

Large games (equivalently, small-influence games) are studied in
Kalai~\cite{E04} and Azrieli and Shmaya~\cite{AS13}.
In these papers, the existence of pure $\epsilon$-Nash equilibria for
$\epsilon = \gamma \sqrt{8n\log(2kn)}$ is established, where $\gamma$
is the largeness/Lipschitz parameter of the game, and $k$ is the number
of pure strategies for each player.
In particular, since we assume that $\gamma = \frac{1}{n}$ and $k=2$ we notice
that $\epsilon = O(n^{-1/2})$ so that there exist arbitrarily
accurate pure Nash equilibria in large games as the number of players
increases. Kearns et al.~\cite{KPRR15} study this class of games from
the mechanism design perspective of mediators who aim to achieve a
good outcome to such a game via recommending actions to players.
Babichenko~\cite{Bab-GEB13} studies large binary-action {\em anonymous} games.
Anonymity is exploited to create
a randomised dynamic on pure strategy profiles that with high
probability converges to a pure approximate
equilibrium in $O(n \log n)$ steps.

Payoff query complexity has been recently studied as a measure of the difficulty
of computing game-theoretic solutions, for various classes of games.
Upper and lower bounds on query complexity have been obtained
for bimatrix games~\cite{FGGS,FS}, congestion games~\cite{FGGS},
and anonymous games~\cite{GT14}.
For general $n$-player games (where the payoff function is exponential
in $n$), the query complexity is exponential in $n$ for exact Nash,
also exact correlated equilibria~\cite{HN13}; likewise for approximate
equilibria with deterministic algorithms
(see also \cite{BB13}). For randomised algorithms,
query complexity is exponential for {\em well-supported} approximate
equilibria~\cite{B13}, which has since been strengthened to
any $\epsilon$-Nash equilibria~\cite{CCT15}.
With randomised algorithms, the query complexity of approximate
correlated equilibrium is $\Theta(\log n)$ for any positive $\epsilon$
\cite{GR14}.

Our main result applies in the setting of {\em completely uncoupled dynamics} in equilibria computation. These dynamics have been studied extensively: Hart and Mas-Colell~\cite{HMC00} show that there exist finite-memory uncoupled strategies that lead to pure Nash equilibria in every game where they exist. Also, there exist finite memory uncoupled strategies that lead to $\epsilon$-NE in every game. Young's interactive trial and error \cite{Y09} outlines completely uncoupled strategies that lead to pure Nash equilibria with high probability when they exist. Regret testing from Foster and Young \cite{FY06} and its $n$-player extension by Germano and Lugosi in \cite{GL05} show that there exist completely uncoupled strategies that lead to an $\epsilon$-Nash equilibrium with high probability. Randomisation is essential in all of these approaches, as Hart and Mas-Colell~\cite{HMC03} show that it is impossible to achieve convergence to Nash equilibria for all games if one is restricted to deterministic uncoupled strategies. This prior work is not concerned with rate of convergence; by contrast here we obtain efficient bounds on runtime.
Convergence in adaptive dynamics for exact Nash equilibria is also studied by Hart and Mansour in \cite{HM10} where they provide exponential lower bounds via communication complexity results. Babichenko~\cite{B13} also proves an exponential lower bound on the rate of convergence of adaptive dynamics to an approximate Nash equilibrium for general binary games. Specifically, he proves that there is no $k$-queries dynamic that converges to an $\epsilon$-WSNE in $\frac{2^{\Omega(n)}}{k}$ steps with probability of at least $2^{-\Omega(n)}$ in all $n$-player binary games. Both of these results motivate the study of specific subclasses of these games, such as the ``large'' games studied here.

%% file: upload-section-largegames-arxiv.tex
\section{Preliminaries}\label{sec:prelim}

We consider games with $n$ players where each player has $k$ actions
$\cA = \{0,1,...,k-1\}$.
Let $a = (a_i, a_{-i})$ denote an~\emph{action
profile} in which player $i$ plays action $a_i$ and the remaining
players play action profile $a_{-i}$. We also consider~\emph{mixed
strategies}, which are defined by the probability distributions over
the action set $\cA$. We write $p = (p_i, p_{-i})$ to denote
a~\emph{mixed-strategy profile} where $p_i$ is a distribution over $\mathcal{A}$ corresponding to the $i$-th player's mixed strategy. To be more precise, $p_i$ is a vector $ (p_{ij})_{j=1}^{k-1}$ such that $\sum_{j=1}^{k-1}p_{ij} \leq 1$ where $p_{ij}$ denotes the $i$-th player's probability mass on her $j$-th strategy. Furthermore, we denote $p_{i0} = 1 - \sum_{j=1}^{k-1}p_{ij}$ to be the implicit probability mass the $i$-th player places on her $0$-th pure strategy.

Each player $i$ has a payoff function
$u_i\colon \cA^n \rightarrow [0, 1]$ mapping an action profile to some
value in $[0,1]$. We will sometimes write
$u_i(p) = \Ex{a\sim p}{u_i(a)}$ to denote the expected payoff of
player $i$ under mixed strategy $p$. An action $a$ is player
$i$'s~\emph{best response} to mixed strategy profile $p$ if
$a \in \argmax_{j\in \mathcal{A}} u_i(j , p_{-i})$.

We assume our algorithms or the players have no other prior knowledge of
the game but can access payoff information through querying
a~\emph{payoff oracle} $\cQ$. For each~\emph{payoff query}
specified by an action profile $a\in \cA^n$, the query oracle will
return $(u_i(a))_{i=1}^n$, the $n$-dimensional vector of payoffs to
each player. Our goal is to compute an~\emph{approximate Nash equilibrium} with a small number of queries.
In the completely uncoupled setting, a query works as follows:
each player $i$ chooses her own action $a_i$ independently of the other
players, and learns her own payoff $u_i(a)$ but no other payoffs.

\begin{definition}[Regret; (approximate) Nash equilibrium]
Let $p$ be a mixed strategy profile, the~\emph{regret} for player $i$ at $p$ is
\[
  reg(p, i) = \max_{j\in\mathcal{A}} \Ex{a_{-i}\sim p_{-i}}{u_i(j,
    a_{-i})} - \Ex{a \sim p}{u_i(a)}.
\]
A mixed strategy profile $p$ is an $\eps$-\emph{approximate Nash
  equilibrium} ($\epsilon$-NE) if for each player $i$, the regret satisfies
$reg(p, i) \leq \eps$.
\end{definition}

In section \ref{sec:various-lipschitz} we will address the stronger notion of a {\em well-supported} approximate Nash equilibrium. In essence, such an equilibrium is a mixed-strategy profile where players only place positive probability on actions that are approximately optimal. In order to precisely define this, we introduce $supp(p_i) = \{j \in \mathcal{A} \ | \ p_{ij} > 0 \}$ to be the set of actions that are played with positive probability in player $i$'s mixed strategy $p_i$.

\begin{definition}[Well-supported approximate Nash equilibrium]\label{def:WSNE}
A mixed-strategy profile $p = (p_i)_{i=1}^n$ is an $\epsilon$ well-supported Nash Equilibrium ($\epsilon$ -WSNE) if and only if the following holds for all players $i \in [n]$:
$$
j \in supp(p_i) \Rightarrow \max_{\ell\in\mathcal{A}} \Ex{a_{-i}\sim p_{-i}}{u_i(\ell,
    a_{-i})} - u_i(j) < \epsilon
$$

\end{definition}

An $\epsilon$-WSNE is always an $\epsilon$-NE, but the converse is not necessarily true as a player may place probability mass on strategies that are more than $\epsilon$ from optimal yet still maintain a low regret in the latter. 

\begin{observation}
To find an exact Nash (or even, correlated)
equilibrium of a large game, in the worst case it is necessary to query the game
exhaustively, even with randomised algorithms. This uses a
similar negative result for general games due to~\cite{HN13},
and noting that we can obtain a strategically equivalent $\gamma$-large game
(Def.~\ref{def:gammalarge}), by scaling down the payoffs into the interval $[0,\gamma]$.
\end{observation}

We will assume the following~\emph{largeness} condition in our games.
Informally, such largeness condition implies that no single
player has a large influence on any other player's utility function.
\begin{definition}[Large Games]\label{def:gammalarge}
  A game is~\emph{$\gamma$-large} if for any two distinct players $i\neq j$,
  any two distinct actions $a_j$ and $a_j'$ for player $j$, and any
  tuple of actions $a_{-j}$ for everyone else:
\[
 |u_i(a_j, a_{-j}) - u_i(a_j', a_{-j})| \leq \gamma \in [0,1].
\]
\end{definition}
We will call $\gamma$ the \emph{largeness parameter} of the game;
in \cite{AS13} this quantity is called the Lipschitz value of the game.
One immediate implication of the largeness assumption is the following
Lipschitz property of the utility functions.

\begin{restatable}{lemma}{lips}\label{lem:lips}
  For any player $i\in [n]$, and any action $j\in \cA$, the fixed
  utility function $u_i(j, p_{-i}):[0,1]^{(n-1)\times (k-1)} \rightarrow [0,1]$ is a $\gamma$-Lipschitz function of the
  second argument $p_{-i} \in [0,1]^{(n-1)\times (k-1)} $ w.r.t. the $\ell_1$ norm.
\end{restatable}

\begin{proof}
  Without loss of generality consider $i = 1$ and $j = 0$. Let
  $q = p_{-1}$ and $q' = p'_{-1}$ be two mixed strategy profiles for
  the other players. For $i \geq 2$ and $j \in \mathcal{A} \setminus \{0\}$, let $\delta_{ij} = q'_{ij} - q_{ij}$. Note that $\|q - q'\|_1 = \sum_{ij} |\delta_{ij}|$.

Let $e_{ij}$ be the unit vector that has a 1 in the $(ij)$-th entry and 0 elsewhere. We first show that there exists an ordering of the discrete set $\{ (ij) \ | \ 2 \leq i \leq n,  \ 1 \leq j \leq k \}$ denoted by $\{\alpha_1, \alpha_2,..., \alpha_{(n-1)(k-1)}\}$ such that for all $\ell = 1,...,(n-1)(k-1)$, the vector $q_\ell = q + \sum_{i=1}^\ell \delta_{\alpha_i} e_{\alpha_i}$ represents valid mixed strategy profiles for players $i \geq 2$.

Suppose that we fix $i$, and consider $q_i$ and $q'_i $ as the mixed strategies of player $i$ arising in $q$ and $q'$. We recall that these are vectors in $[0,1]^{k-1}$ whose components sum is less than 1. We consider two cases. In the first, suppose that there exists a $j$ such that $\delta_{ij} < 0$ by definition, $\delta_{ij} < q_{ij}$, hence $q_i + \delta_{ij} e_j$ is a valid mixed strategy for player $i$.

In the second, suppose that $\delta_{ij} > 0$ for all $j$. Now suppose that $\delta_{ij} > q_{i0} = 1 - \sum_{j=1}^{k-1} q_{ij}$ for all $j$. If such is the case then $q'_{i}$ cannot possibly be a valid mixed strategy for player $i$, hence it must be the case that for some $j$, $\delta_{ij} < q_{i0}$, hence once again $q_i + \delta_{ij} e_j$ is a valid mixed strategy for player $i$. 

Since such a choice of valid updates by $\delta_{ij}$ can always be found for valid $q_i$ and $q'_i$, we can recursively find valid shifts by $\delta_{ij}$ in a specific coordinate to reach $q'_i$ from $q_i$. If this is applied in order for all players $i \geq 2$, the aforementioned claim holds and indeed $q_\ell = q + \sum_{i=1}^\ell \delta_{\alpha_i} e_{\alpha_i}$ for some ordering $\{\alpha_1,...,\alpha_{(n-1)(k-1)}\}$. 

With this in hand, we can use telescoping sums and the largeness condition to prove our lemma. For simplicity of notation, in what follows we assume that $q_0 = q$, and we recall that by definition $q_{(n-1)(k-1)} = q'$.

\begin{align*}
  |u_i(j, q') - u_i(j, q)| &=   \left|\sum_{\ell = 1}^{(n-1)(k-1)} u_i(j, q_\ell) - u_i(j, q_{\ell-1}) \right| \\
  \mbox{(Triangle Inequality)}\qquad &\leq \sum_{\ell = 1}^{(n-1)(k-1)} \left| u_i(j, q_\ell) - u_i(j, q_{\ell-1}) \right|\\
  \mbox{(Definition of Largeness)}\qquad &\leq \sum_{\ell = 1}^{(n-1)(k-1)} \gamma |\delta_{\alpha_\ell}| = \gamma \|q' - q\|_1
\end{align*}
which proves our claim.
\end{proof}

From now on until Section~\ref{sec:extensions} we will focus on $\frac{1}{n}$-large
binary action games where $\mathcal{A} = \{0,1\}$ and $\gamma =\frac{1}{n}$.
The reason for this is that the techniques we introduce can be more conveniently
conveyed in the special case of $\gamma=\frac{1}{n}$, and subsequently extended
to general $\gamma$.

Recall that $p_i$ denotes a mixed strategy of player $i$.
In the special case of binary-action games, we slightly abuse the notation
to let $p_i$ denote the probability that player $i$ plays 1 (as opposed to
0), since in the binary-action case, this single probability describes $i$'s
mixed strategy.

The following notion of~\emph{discrepancy} will be useful.
\begin{definition}[Discrepancy]
Letting $p$ be a mixed strategy profile, the~\emph{discrepancy} for player $i$ at $p$ is
\[
disc(p, i) = \left| \Ex{a_{-i}\sim p_{-i}}{u_i(0, a_{-i})} - \Ex{a_{-i}\sim p_{-i}}{u_i(1, a_{-i})} \right|.
\]
\end{definition}

\begin{paragraph}{Estimating payoffs for mixed profiles}
We can approximate the expected payoffs for any mixed strategy profile
by repeated calls to the oracle $\cQ$. In particular, for any target
accuracy parameter $\beta$ and confidence parameter $\delta$, consider
the following procedure to implement an oracle $\cQ_{\beta, \delta}$:
\begin{itemize}
\item For any input mixed strategy profile $p$, compute a new mixed
  strategy profile $p' = (1 - \frac{\beta}{2})p + (\frac{\beta}{2})\mathbf{1}$ such
  that each player $i$ is playing uniform distribution with
  probability $\frac{\beta}{2}$ and playing distribution $p_i$ with
  probability $1 - \frac{\beta}{2}$.
\item Let $N = \frac{64}{\beta^3} \log\left( 8n /\delta \right)$, and
  sample $N$ payoff queries randomly from $p'$, and call the oracle
  $\cQ$ with each query as input to obtain a payoff vector.
\item Let $\hat u_{i,j}$ be the average sampled payoff to player $i$
  for playing action $j$.\footnote{If the player $i$ never plays an
    action $j$ in any query, set $\hat u_{i,j} = 0$.} Output the
  payoff vector $(\hat{u}_{ij})_{i\in [n], j\in\{0, 1\}}$.
\end{itemize}
\end{paragraph}

\begin{lemma}\label{lem:approx-query}
  For any $\beta, \delta \in (0, 1)$ and any mixed strategy profile
  $p$, the oracle $\cQ_{\beta, \delta}$ with probability at least
  $1 - \delta$ outputs a payoff vector
  $(\hat u_{i,j})_{i\in [n], j\in\{0, 1\}}$ that has an additive error of at
  most $\beta$, that is for each player $i$, and each action
  $j\in \{0, 1\}$,
\[
  |u_i(j, p_{-i}) - \hat{u}_{i, j}| \leq \beta.
\]
\end{lemma}

The lemma follows from Proposition 1 of~\cite{GR14} and the largeness property.

\begin{paragraph}{Extension to Stochastic Utilities.}
We consider a generalisation
where the utility to player $i$ of any pure profile $a$ may consist of a probability
distribution $D_{a,i}$ over $[0,1]$, and if $a$ is played, $i$ receives a
sample from $D_{a,i}$.
The player wants to maximise her expected utility with respect to
sampling from a (possibly mixed) profile, together with sampling from
any $D_{a,i}$ that results from $a$ being chosen.
If we extend the definition of $\cQ$ to output samples of the $D_{a,i}$ for
any queried profile $a$, then $\cQ_{\beta,\delta}$ can be defined in
a similar way as before, and simulated as above using samples from $\cQ$.
Our algorithmic results extend to this setting.
\end{paragraph}


%% file: upload-section-warmup-arxiv.tex
\section{Warm-up: $0{\cdot}25$-Approximate Equilibrium}\label{sec:warmup}

In this section, we exhibit some simple procedures whose general approach is
to query a constant number of mixed strategies (for which additive approximations
to the payoffs can be obtained by sampling).
Observation~\ref{obs:half} notes that a $\frac{1}{2}$-approximate Nash equilibrium
can be found without using any payoff queries:

\begin{observation}\label{obs:half}
Consider the following ``uniform''
mixed strategy profile. Each player puts $\frac{1}{2}$ probability mass on
each action: for all $i$, $p_{i} = \frac{1}{2}$. Such a mixed strategy
profile is a $\frac{1}{2}$-approximate Nash equilibrium.
\end{observation}

We present two algorithms that build on Observation~\ref{obs:half} to
obtain better approximations than $\frac{1}{2}$.
For simplicity of presentation,
we assume that we have access to a mixed strategy query oracle
$\cQ_M$ that returns exact expected payoff values for any input mixed
strategy $p$. Our results continue to hold if we replace $\cQ_M$ by
$\cQ_{\beta, \delta}$.~\footnote{In particular, if we use $\cQ_{\beta,\delta}$
for our query access, then with probability at least $1-\delta$
we will get $(\eps + O(\beta))$-approximate equilibrium,
where $\eps$ is the approximation performance obtainable via access to $\cQ_M$.}

\paragraph{\bf Obtaining $\eps= 0{\cdot}272$.}
{
First, we show that having each player making small adjustment from the
``uniform'' strategy can improve $\eps$ from $\frac{1}{2}$ to around
$0{\cdot}27$. We simply let players with large regret shift more
probability weight towards their best responses. More formally,
consider the following algorithm {\bf OneStep} with two parameters
$\alpha, \Delta\in [0, 1]$:
\begin{itemize}
\item Let the players play the ``uniform'' mixed strategy. Call the
  oracle $\cQ_{M}$ to obtain the payoff values of $u_i(0, p_{-i})$ and
  $u_i(1, p_{-i})$ for each player $i$.
\item For each player $i$, if
  $u_i(0, p_{-i}) - u_i(1, p_{-i}) > \alpha$, then set $p_{i} = \frac{1}{2} - \Delta$;
  if $u_i(1, p_{-i}) - u_i(0, p_{-i}) > \alpha$, set $p_{i} = \frac{1}{2} + \Delta$;
  otherwise keep playing $p_{i} = \frac{1}{2}$.
\end{itemize}

\begin{restatable}{theorem}{onestep}
\label{lem:onestep}
If we use algorithm {\bf OneStep} with parameters $\alpha = 2 - \sqrt{\frac{11}{3}}$ and
$\Delta = \sqrt{\frac{11}{48}} - \frac{1}{4}$, then the resulting mixed strategy profile is an
$\eps$-approximate Nash equilibrium with $\eps \leq 0{\cdot}272$.
\end{restatable}
}

\begin{proof}
  Let $p$ denote the ``uniform'' mixed strategy, and $p'$ denote the
  output strategy by {\bf OneStep}. We know that
  $\|p - p'\|_1 \leq n \Delta$. By~\Cref{lem:lips}, we know that for
  any player $i$ and action $j$,
  $|u_i(j, p_{-i}) - u_i(j, p_{-j}')|\leq \Delta$.

  Consider a player $i$ whose discrepancy in $p$ satisfies
  $disc(p,i) \leq \alpha$. Then such player's discrepancy in $p'$ is
  at most $disc(p', i)\leq \alpha + 2\Delta$, so her regret in $p'$ is
  bounded by
\begin{equation}\label{junk}
    reg(p', i) = p_i' \, disc(p', i) = disc(p', i)/2 \leq \alpha/2 +
    \Delta.
\end{equation}

Consider a player $i$ such that $disc(p,i)>\alpha$. Then we consider
two different cases. In the first case, the best response of player
$i$ remains the same in both profiles $p$ and $p'$. Since
$disc(p', i) \leq 1$, we can bound the regret by
\begin{equation}\label{junk1}
reg(p', i) = p_i' \, disc(p', i) = \left( \frac{1}{2} - \Delta \right).
\end{equation}

In the second case, the best response of player $i$ changes when the
profile $p$ changes to $p'$. In this case, the discrepancy is at most
$2\Delta - \alpha$, and so the regret is bounded by
\begin{equation}\label{junk2}
  reg(p', i) = p_i'\, disc(p', i) = \left( \frac{1}{2} + \Delta \right)(2\Delta - \alpha).
\end{equation}
By combining all cases from~\Cref{junk,junk1,junk2}, we know the
regret is upper-bounded by
\begin{equation}\label{junk3}
  reg(p', i) \leq \max \left( \frac{\alpha}{2}+ \Delta,
    \frac{1}{2}-\Delta, \frac{1}{2} (1 + 2\Delta) (2\Delta - \alpha)
  \right)
\end{equation}

By choosing values
\[
(\alpha^*, \Delta^*) = \left( 2 - \sqrt{\frac{11}{3}}, \sqrt{\frac{11}{48}} - \frac{1}{4}  \right) \approx (0{\cdot}085, 0{\cdot}229)
\]
The right hand side of~\Cref{junk3} is bounded by $0{\cdot}272$.  Thus if we
use the optimal $\alpha^*$ and $\Delta^*$ in our algorithm, we can
attain an $\epsilon = 0{\cdot}272$ approximate Nash equilibrium.
\end{proof}

\paragraph{\bf Obtaining $\epsilon = 0{\cdot}25$.}
{
We now give a slightly more sophisticated algorithm than the previous
one. We will again have the players starting with the ``uniform''
mixed strategy, then let players shift more weights toward their best
responses, and finally let some of the players switch back to the
uniform strategy if their best responses change in the
adjustment. Formally, the algorithm {\bf TwoStep} proceeds as:
\begin{itemize}
\item Start with the ``uniform'' mixed strategy profile, and query the
  oracle $\cQ_M$ for the payoff values. Let $b_i$ be player $i$'s best
  response.
\item For each player $i$, set the probability of playing their best
  response $b_i$ to be $\frac{3}{4}$. Call $\cQ_M$ to obtain payoff values for
  this mixed strategy profile, and let $b'_i$ be each player $i$'s
  best response in the new profile.
\item For each player $i$, if $b_i\neq b'_i$, then resume playing
  $p_{i} = \frac{1}{2}$. Otherwise maintain the same mixed
  strategy from the previous step.
\end{itemize}

\begin{restatable}{theorem}{twostep}\label{lem:twostep}
  The mixed strategy profile output by {\bf TwoStep} is an
  $\eps$-approximate Nash equilibrium with $\eps \leq 0{\cdot}25$.
\end{restatable}

}

\begin{proof}
  Let $p$ denote the ``uniform'' strategy profile, $p'$ denote the
  strategy profile after the first adjustment, and $p''$ denote the
  output strategy profile by {\bf TwoStep}.

  For any player $i$, there are three cases regarding the discrepancy
  $disc(p, i)$.
\begin{enumerate}
\item The discrepancy $disc(p, i) > \frac{1}{2}$;
\item The discrepancy $disc(p, i) \leq \frac{1}{2}$ and player $i$ returns to the
  uniform mixed strategy at the end;
\item The discrepancy $disc(p, i) \leq \frac{1}{2}$ and player $i$ does not
  return to the uniform mixed strategy in the end. 
\end{enumerate}

Before we go through all the cases, the following facts are
useful. Observe that $\|p - p'\|, \|p - p''\|, \|p' - p''\|\leq n/4$,
so for any action $j$,
\begin{equation}\label{ohman}
\max\{|u_i(j, p_{-i}') - u_i(j, p_{-i}'')|, |u_i(j, p_{-i}) - u_i(j,
p_{-i}'), |u_i(j, p_{-i}) - u_i(j, p_{-i}'')| \}
\leq \frac{1}{4}
\end{equation}
It follows that
\[
  \max\{|disc(p',i) - disc(p'', i)|, |disc(p, i) - disc(p', i)|,
  |disc(p,i ) - disc(p'', i)| \} \leq \frac{1}{2}
\]
We will now bound the regret of player $i$ in the first case.  Since
in the mixed strategy profile $p$, the best response of player $i$ is
better than the other action by more than $\frac{1}{2}$. This means the best
response action will remain the same in $p'$ and $p''$ for this
player, and she will play this action with probability $\frac{3}{4}$ in the end,
so her regret is bounded by $\frac{1}{4}$.

Let us now focus on the second case where discrepancy
$disc(p, i) \leq \frac{1}{2}$ and player $i$ returns to the uniform strategy
of part 1. It is sufficient to show that the discrepancy at the end
satisfies $disc(p'', i) \leq \frac{1}{2}$. Without loss generality, assume
that the player best response in the ``uniform'' strategy profile is
action $b_i = 1$, and the best response after the first adjustment is
action $b_i = 0$. This means
\[
u_i(1, p_{-i}) - u_i(0, p_{-i})\geq 0
\quad \mbox{and, }\quad
u_i(0, p_{-i}') - u_i(1, p_{-i}') \geq 0.
\]
By combining with~\Cref{ohman}, we have
\begin{align*}
  u_i(1, p_{-i}'') - u_i(0, p_{-i}'')\leq  u_i(1, p_{-i}') - u_i(0, p_{-i}') + \frac{1}{2}  \leq  \frac{1}{2}\\
  u_i(0, p_{-i}'') - u_i(1, p_{-i}'') \leq  u_i(0, p_{-i}) - u_i(1, p_{-i}) + \frac{1}{2} \leq \frac{1}{2}.
\end{align*}
Therefore, we know $disc(p'',i) \leq \frac{1}{2}$, and hence the regret
$reg(p'',i) \leq \frac{1}{4}$.

Finally, we consider the third case where
$disc(p, i) \leq \frac{1}{2}$ and player $i$ does not return to a
uniform strategy. Without loss generality, assume that action 1 is
best response for player $i$ in both $p$ and $p'$, and so
$u_i(1, p'_{-i}) \geq u_i(0, p'_{-i})$. By~\Cref{ohman}, we also have
\[
u_i(0, p_{-i}'') - u_i(1, p_{-i}'') \leq \frac{1}{2}.
\]
If in the end her best response changes to 0, then the regret is
bounded by $reg(p'',i) \leq \frac{1}{8}$. Otherwise if the best
response remains to be 1, then the regret is again bounded by
$reg(p'', i) \leq \frac{1}{4}$

Hence, in all of the cases above we could bound the player's regret by
$\frac{1}{4}$.\end{proof}


%% file: upload-section-continuous-dynamics-arxiv.tex
\section{$\frac{1}{8}$-Approximate Equilibrium via Uncoupled Dynamics}

In this section, we present our main algorithm that achieves
approximate equilibria with $\eps \approx \frac{1}{8}$ in a completely
uncoupled setting. In order to arrive at this we first model game
dynamics as an uncoupled continuous-time dynamical system where a
player's strategy profile updates depend only on her own
mixed strategy and payoffs. Afterwards we present a discrete-time
approximation to these continuous dynamics to arrive at a query-based
algorithm for computing $(\frac{1}{8} + \alpha)$-Nash equilibrium
with query complexity logarithmic in the number of players.
Here, $\alpha>0$ is a parameter that can be chosen, and the number of mixed-strategy
profiles that need to be tested is inversely proportional to $\alpha$.
Finally, as mentioned in Section~\ref{sec:prelim}, we recall that these
algorithms carry over to games with stochastic utilities, for which we can
show that our algorithm uses an essentially optimal number of queries.

Throughout the section, we will rely on the following notion
of a~\emph{strategy/payoff state}, capturing the information
available to a player at any moment of time.

\begin{definition}[Strategy-payoff state]
  For any player $i$, the~\emph{strategy/payoff state} for player $i$
  is defined as the ordered triple $s_i = (v_{i1}, v_{i0}, p_i)\in [0, 1]^3$,
  where $v_{i1}$ and $v_{i0}$ are the player's utilities for playing
  pure actions 1 and 0 respectively, and $p_i$ denotes the player's
  probability of playing action 1. Furthermore, we denote the player's
  discrepancy by $D_i = |v_{i1} - v_{i0}|$ and we let $p_i^*$ denote the
  probability mass on the best response, that is if
  $v_{i1} \geq v_{i0}$, $p_i^* = p_i$, otherwise $p_i^* = 1-p_i$.
\end{definition}

\subsection{Continuous-Time Dynamics}

First, we will model game dynamics in continuous time, and assume that
a player's strategy/payoff state (and thus all variables it contains)
is a differentiable time-valued function. When we specify these values
at a specific time $t$, we will write
$s_i(t) = (v_{i1}(t), v_{i0}(t), p_i(t))$. Furthermore, for any
time-differentiable function $g$, we denote its time derivative by
$\dot{g} = \frac{d}{dt} g$.  We will consider continuous game dynamics
formally defined as follows.

\begin{definition}[Continuous game dynamic]
  A continuous game dynamic consists of an update function $f$ that specifies a player's strategy update at time $t$. Furthermore, $f$ depends only on $s_i(t)$ and $\dot{s}_i(t)$. In other words,
  $\dot{p}_i(t) = f(s_i(t), \dot{s}_i(t))$ for all $t$.
\end{definition}

\begin{observation}
We note that in this framework, a specific player's updates do not depend on
other players' strategy/payoff states nor their history of play.
This will eventually lead us to uncoupled Nash equilibria computation
in Section~\ref{sec:discrete-time-step}.
\end{observation}

A central object of interest in our continuous dynamic is a linear
sub-space $\mathcal{P} \subset [0,1]^3$ such that all strategy/payoff
states in it incur a bounded regret.  Formally, we will define
$\mathcal{P}$ via its normal vector $\vec{n} = (-\frac{1}{2},\frac{1}{2},1)$ so that
$\mathcal{P} = \{ s_i | \ s_i \cdot \vec{n} = \frac{1}{2} \}$. Equivalently,
we could also write
$\mathcal{P} = \{s_i \ | \ p_i^* =
\frac{1}{2}(1+D_i)\}$. (See~\Cref{fig:plane} for a visualisation.)
With this observation, it is straightforward to see that any player
with strategy/payoff state in $\mathcal{P}$ has regret at most $\frac{1}{8}$.

\begin{figure}
\center{
\begin{tikzpicture}[scale=0.7]
\tikzstyle{xxx}=[dashed,thick]


\draw[thick,<->](5,0)--(0,2)--(0,9); \draw[thick,->](0,2)--(12,3); 
\draw[xxx](6,8.5)--(6,2.5)--(9.5,1.7);

\fill[red!20](3,0.8)--(0,5)--(6,8.5)--(9.5,4.75)--cycle; 
\draw[thick,red](0,5)--(9.5,4.75);
\draw[xxx](0,2)--(12,3); 
\draw[xxx,gray](6,2.5)--(6,8.5); 

\draw[thick](3,0.8)--(9.5,1.7)--(9.5,8)--(3,7.5)--cycle; 
\draw[thick](3,7.5)--(0,8)--(6,8.5)--(9.5,8); 

\node at(-0.5,2){$0$};\node at(-0.5,5){$\frac{1}{2}$};\node at(-0.5,8){$1$};\node at(-2.5,8.7){$p_i=\Pr[{\rm play}~1]$};
\node at(5.25,0.25){$v_{i0}$};\node at(12.3,3){$v_{i1}$};

\node at(3,0.8){\textcolor{red}{$\bullet$}};\node at(6,8.5){\textcolor{red}{$\bullet$}};
\end{tikzpicture}
\caption{Visualisation of $\mathcal{P}$; on the red line,
$v_{i0}=v_{i1}$ so the player is indifferent and mixes with equal probabilities;
at the red points the player has payoffs of 0 and 1, and makes
a pure best response.}\label{fig:plane}}
\end{figure}
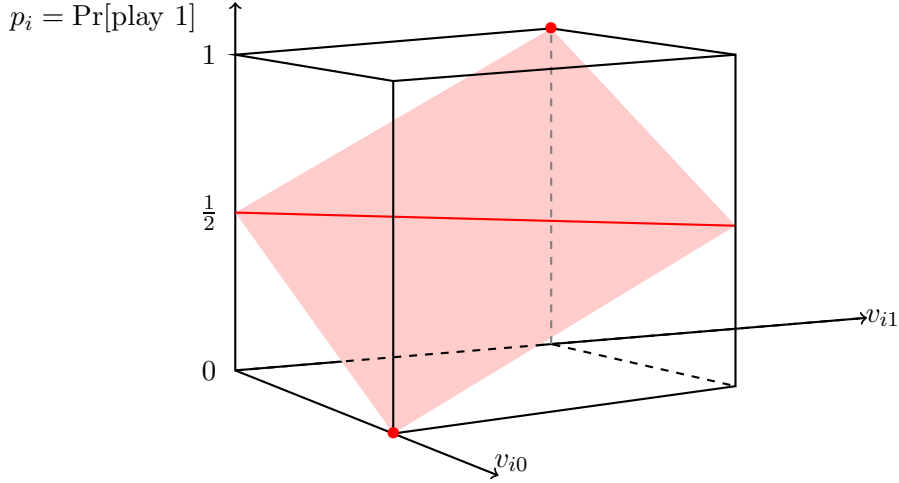

\begin{lemma}\label{lem:max-regret}
If player $i$'s strategy/payoff state satisfies $s_i\in \cP$,
then her regret is at most $\frac{1}{8}$.
\end{lemma}

\begin{proof}
  This follows from the fact that a player's regret can be expressed
  as $D_i(1-p_i^*)$ and the fact that all points on $\mathcal{P}$ also
  satisfy $p_i^* = \frac{1}{2}(1+D_i)$.  In particular, the maximal
  regret of $\frac{1}{8}$ is achieved when $D_i = \frac{1}{2}$ and
  $p_i^* = \frac{3}{4}$.
\end{proof}

Next, we want to show there exists a dynamic that allows all players to
eventually reach $\cP$ and remain on it over time. We notice that
for a specific player, $\dot{v}_{i1}$, $\dot{v}_{i0}$ and subsequently
$\dot{D}_i$ measure the cumulative effect of other players shifting
their strategies. However, if we limit how much any individual player
can change their mixed strategy over time by imposing
$|\dot{p}_i| \leq 1$ for all $i$, ~\Cref{lem:lips} guarantees
$|\dot{v}_{ij}| \leq 1$ for $j = 0,1$ and consequently
$|\dot{D}_i| \leq 2$. With these quantities bounded, we can
consider an adversarial framework where we construct game dynamics by
solely assuming that $|\dot{p}_i(t)| \leq 1$,
$|\dot{v}_{ij}(t)| \leq 1$ for $j = 0,1$ and $|\dot{D}_i(t)| \leq 2$
for all times $t\geq 0$.

Now assume an adversary controls $\dot{v}_{i0}$, $\dot{v}_{i1}$ and
hence $\dot{D}_i$, one can show that if a player sets 
$\dot{p}_i(t) = \frac{1}{2}(\dot{v}_{i1}(t) - \dot{v}_{i0}(t))$, then
she could stay on $\cP$ whenever she reaches the subspace.

\begin{lemma}\label{lem:remain-plane}
If $s_i(0) \in \mathcal{P}$, and $\dot{p}_i(t) = \frac{1}{2}(\dot{v}_{i1}(t) - \dot{v}_{i0}(t))$, then $s_i(t) \in \mathcal{P} \ \forall \ t \geq 0$.  
\end{lemma}

\begin{theorem}\label{thm:UNC}
Under the initial conditions $p_i(0) = \frac{1}{2}$ for all $i$, the following continuous dynamic, \textbf{Uncoupled Continuous Nash (UCN)}, has all players reach $\mathcal{P}$ in at most $\frac{1}{2}$ time units. Furthermore, upon reaching $\mathcal{P}$ a player never leaves.
\[
 \dot{p}_i(t) = f(s_i(t), \dot{s}_i(t)) = 
  \begin{cases} 
      \hfill \ \ 1    \hfill & \text{ if } s_i \notin \mathcal{P} \text{ and } v_{i1} \geq v_{i0} \\
      \hfill -1 \hfill & \text{ if } s_i \notin \mathcal{P} \text{ and } v_{i1} < v_{i0} \\
      \hfill \frac{1}{2}(\dot{v}_{i1}(t) - \dot{v}_{i0}(t)) \hfill & \text{ if } s_i \in  \mathcal{P}
  \end{cases}
\]

\end{theorem}

\begin{proof}
  From \Cref{lem:remain-plane} it is clear that once a player reaches
  $\mathcal{P}$ they never leave the plane. It remains to show that it
  takes at most $\frac{1}{2}$ time units to reach $\mathcal{P}$.

  Since $p_i(0) = p_i^*(0) = \frac{1}{2}$, it follows that if $s_i(0)
  \notin \mathcal{P}$ then $p_i^*(0) < \frac{1}{2}(1 + D_i(0))$. On
  the other hand, if we assume that $\dot{p}_i^*(t) = 1$ for $t \in 
  [0,\frac{1}{2}]$, and that player preferences do not change,
  then it follows that $p_i^*(\frac{1}{2}) = 1$ and $p_i^*(\frac{1}{2})
  \geq \frac{1}{2}(1 + D_i(\frac{1}{2}))$, where equality holds only if
  $D_i(\frac{1}{2}) = 1$. By continuity of $p_i^*(t)$ and $D_i(t)$ it
  follows that for some $k \leq \frac{1}{2}$, $s_i(k) \in \mathcal{P}$.
  It is simple to see that the same holds in the case where preferences change.
\end{proof}

\subsection{Discrete Time-step Approximation}\label{sec:discrete-time-step}

The continuous-time dynamics of the previous section hinge on
obtaining expected payoffs in mixed strategy profiles, thus we will
approximate expected payoffs via $\mathcal{Q}_{\beta,\delta}$.
Our algorithm will have each player adjusting their mixed strategy over
rounds, and in each round query $\cQ_{\beta, \delta}$ to obtain the payoff values.

Since we are considering discrete approximations to~\textbf{UCN},
the dynamics will no longer guarantee that strategy/payoff states stay
on the plane $\mathcal{P}$. For this reason we define the following
region around $\mathcal{P}$:

\begin{definition}
Let $\mathcal{P}^\lambda = \{ s_i \ | \ s_i \cdot \vec{n} \in [\frac{1}{2} - \lambda, \frac{1}{2} + \lambda]\}$, with normal vector $\vec{n} = (-\frac{1}{2},\frac{1}{2},1)$. Equivalently, $\mathcal{P}^\lambda = \{s_i \ | \ p_i^* = \frac{1}{2}(1+D_i) + c, \  c \in[-\lambda,\lambda]\}$.
\end{definition}

\noindent
Just as in the proof of ~\Cref{lem:max-regret}, we can use the fact that a player's regret is $D_i(1-p_i^*)$ to bound regret on $\mathcal{P}^\lambda$.

\begin{lemma}\label{lem:max-regret-approx}
The worst case regret of any strategy/payoff state in $\mathcal{P}^\lambda$ is $\frac{1}{8}(1+2\lambda)^2$. This is attained on the boundary: $\partial \mathcal{P}^\lambda = \{ s_i  \ | \ s_i \cdot \vec{n} = \frac{1}{2} \pm \lambda \}$.
\end{lemma}

\begin{corollary}\label{lambda-for-epsilon}
For a fixed $\alpha > 0$, if $\lambda = \frac{\sqrt{1+ 8\alpha} - 1}{2}$, then $\mathcal{P}^\lambda$ attains a maximal regret of $\frac{1}{8} +\alpha$. 
\end{corollary}

 \newcommand{\un}{\textbf{UN}} We present an algorithm in the
completely uncoupled setting, {\un($\alpha,\eta$)}, that for any
parameters $\alpha, \eta \in (0, 1]$ computes a
$(\frac{1}{8} + \alpha)$-Nash equilibrium with probability at least
$1 - \eta$. 

Since $p_i(t) \in [0,1]$ is the mixed strategy of the $i$-th player at
round $t$ we let $p(t) = (p_i(t))_{i=1}^n$ be the resulting mixed
strategy profile of all players at round $t$. Furthermore, we use the
mixed strategy oracle $\mathcal{Q}_{\beta,\delta}$
from~\Cref{lem:approx-query} that for a given mixed strategy profile
$p$ returns the vector of expected payoffs for all players with an
additive error of $\beta$ and a correctness probability of $1-\delta$.

The following lemma is used to prove the correctness of
{\un($\alpha,\eta$)}: 
\begin{lemma}\label{lem:bounded-step}
  Suppose that $w \in \mathbb{R}^3$ with $\| w \|_\infty \leq \lambda$
  and let function $h(x) = x \cdot \vec{n}$, where $\vec{n}$ is the
  normal vector of $\cP$. Then
  $h(x + w) - h(x) \in [-2\lambda, 2\lambda]$. Furthermore, if
  $w_3 = 0$, then $h(x + w) - h(x) \in [-\lambda, \lambda]$.
\end{lemma}

\begin{proof}
  The statement follows from the following expression:
$$
h(x + w) - h(x) = w \cdot \vec{n} = \frac{1}{2} (w_2 - w_1) + w_3
$$ 
\end{proof}

\begin{algorithm}                      
\caption{{\un($\alpha,\eta$)}}     
\label{approx_nash_block}                           
\begin{algorithmic}                    
    \REQUIRE 
    \STATE Threshold: $\alpha > 0$    
    \STATE Confidence: $\eta > 0$
    \STATE
	\INITIALIZATION
	\STATE $\lambda \leftarrow \frac{\sqrt{1+8\alpha}-1}{2}$ 
	\STATE $\Delta \leftarrow \frac{\lambda}{4}$
	\STATE $N \leftarrow \lceil \frac{2}{\Delta} \rceil$
	\STATE $p_i(0) \leftarrow  \frac{1}{2}$ for $i \in [n]$
	\STATE   
    \item[\textbf{Initial Gradient Estimate:}]
    \FOR{$(i,j) \in [n]\times \{0,1\}$}
    	\STATE $\hat{v}_{ij}(-1) = \left( \mathcal{Q}_{(\Delta,\frac{\eta}{N})}(p(0)) \right)_{i,j}$ 
	\ENDFOR    
    \STATE
    \item[\textbf{Query Dynamics:}]
	\FOR{$t = 1,...,T$}
	\FOR{$(i,j) \in [n] \times \{0,1\}$}
		\STATE $\hat{v}_{ij}(t) \leftarrow \left( \mathcal{Q}_{(\Delta,\frac{\eta}{N})}(p(t)) \right)_{i,j}$
		\STATE $\Delta \hat{v}_{ij}(t) \leftarrow \hat{v}_{ij}(t) - \hat{v}_{ij}(t-1)$		
		\IF{$\hat{s}_i(t) = \left( \hat{v}_{i1}(t),\hat{v}_{i0}(t),p_i(t) \right) \notin \mathcal{P}^{\lambda/4}$}
			\STATE $p_i^*(t+1) \leftarrow p_i^*(t) + \Delta$
		\ELSE
			\STATE $p_i^*(t+1) \leftarrow p_i^*(t) + \frac{1}{2}(\Delta \hat{v}_{i1}(t) - \Delta \hat{v}_{i0}(t))$
		\ENDIF
	\ENDFOR
	\ENDFOR
	\RETURN $p(t)$
  
\end{algorithmic}
\end{algorithm}

\begin{theorem}\label{UN-proof}
With probability $1 - \eta$, {\un($\alpha,\eta$)} correctly returns a $(\frac{1}{8} + \alpha)$-approximate Nash equilibrium by using $O(\frac{1}{\alpha^4} \log \left( \frac{n}{\alpha \eta} 
\right) ) $  queries. 
\end{theorem}

\begin{proof}
  By ~\Cref{lem:approx-query} and union bound, we can guarantee that
  with probability at least $1-\eta$ all sample approximations to
  mixed payoff queries have an additive error of at most
  $\Delta = \frac{\lambda}{4}$. We will condition on this accuracy
  guarantee in the remainder of our argument. Now we can show that for
  each player there will be some round $k \leq N$, such that at the
  beginning of the round their strategy/payoff state lies in
  $\mathcal{P}^{\lambda/2}$. Furthermore, at the beginning of all
  subsequent rounds $t \geq k$, it will also be the case that their
  strategy/payoff state lies in $\mathcal{P}^{\lambda/2}$.

The reason any player generally reaches $\mathcal{P}^{\lambda/2}$ follows from the fact that in the worst case, after increasing $p^*$ by $\Delta$ for $N$ rounds, $p^*=1$, in which case a player is certainly in $\mathcal{P}^{\lambda/2}$. Furthermore, ~\Cref{lem:bounded-step} guarantees that each time $p^*$ is increased by $\Delta$, the value of $\hat{s}_i \cdot \vec{n}$ changes by at most $\frac{\lambda}{2}$ which is why $\hat{s}_i$ are always steered towards $\mathcal{P}^{\lambda/4}$. Due to inherent noise in sampling, players may at times find that $\hat{s}_i$ slightly exit $\mathcal{P}^{\lambda/4}$ but since additive errors are at most $\frac{\lambda}{4}$. We are still guaranteed that true $s_i$ lie in $\mathcal{P}^{\lambda/2}$.

The second half of step 4 forces a player to remain in $\mathcal{P}^{\lambda/2}$ at the beginning of any subsequent round $t \geq k$. The argumentation for this is identical to that of ~\Cref{lem:remain-plane} in the continuous case.

Finally, the reason that individual probability movements are restricted to $\Delta = \frac{\lambda}{4}$ is that at the end of the final round, players will move their probabilities and will not be able to respond to subsequent changes in their strategy/payoff states. From the second part of ~\Cref{lem:bounded-step}, we can see that in the worst case this can cause a strategy/payoff state to move from the boundary of $\mathcal{P}^{\lambda/2}$ to the boundary of $\mathcal{P}^{\frac{3\lambda}{4}} \subset \mathcal{P}^\lambda$. However, $\lambda$ is chosen in such a way so that the worst-case regret within $\mathcal{P}^{\lambda}$ is at most $\frac{1}{8}+\alpha$, therefore it follows that {\un($\alpha,\eta$)} returns a $\frac{1}{8} + \alpha$ approximate Nash equilibrium. Furthermore, the number of queries is
$$
(N+1) \left( \frac{1024}{\lambda^3} \log \left(\frac{8nN}{\eta} \right) \right) 
= \left( \frac{1}{\lambda} + 1 \right) \left( \frac{1024}{\lambda^3} \log \left(\frac{8n}{\lambda \eta} \right) \right).
$$
It is not difficult to see that $\frac{1}{\lambda} = O(\frac{1}{\alpha})$ which implies that the number of queries made is $O \left( \frac{1}{\alpha^4} \log \left( \frac{n}{ \alpha \eta} \right) \right)$ in the limit.
\end{proof}

\subsection{Logarithmic Lower Bound}

As mentioned in the preliminaries section, all of our previous results extend to stochastic utilities. In particular, if we assume that $G$ is a game with stochastic utilities where expected payoffs are large with parameter $\frac{1}{n}$, then we can apply {\un($\alpha,\eta$)} with $O( \log (n))$ queries to obtain a mixed strategy profile where no player has more than $\frac{1}{8} + \alpha$ incentive to deviate. Most importantly, for $\ell>2$, we can use the same methods as \cite{GR14} to lower bound the query complexity of computing a mixed strategy profile where no player has more than $(\frac{1}{2} - \frac{1}{\ell})$ incentive to deviate.

\begin{restatable}{theorem}{lbstochastic}
If $\ell>2$, the query complexity of computing a mixed strategy profile where no player has more than $(\frac{1}{2} - \frac{1}{\ell})$ incentive to deviate for stochastic utility games is $\Omega (\log_{\ell(\ell-1)} (n))$. Alongside ~\Cref{UN-proof} this implies the query complexity of computing mixed strategy profiles where no player has more than $\frac{1}{8}$ incentive to deviate in stochastic utility games is $\Theta (\log (n))$. 
\end{restatable}

\begin{proof}

Suppose that we have $n$ players and that $\ell>2$. For every $b \in \{0,1\}^n$ we can construct a stochastic utility game $G_b$ as follows: For each player $i$, the utility of strategy $b_i$ is bernoulli with bias $\frac{\ell}{\ell-1}$ and the utility of strategy $1-b_i$ is bernoulli with bias $\frac{1}{\ell}$. Note that this game is trivially $\left( \frac{1}{n} \right)$-Lipschitz, as each player's payoff distributions are completely independent of other players' strategies. 

Suppose that $\mathcal{G}$ is the uniform distribution on the set of all $G_b$, then using the same argumentation as Theorem 3 of \cite{GR14}, we get the following:

\begin{theorem}
Let $\mathcal{A}$ be a deterministic payoff-query algorithm that uses at most $\log_{\ell(\ell-1)}(n)$ queries and outputs a mixed strategy $p$. If $\mathcal{A}$ performs on $\mathcal{G}$, then with probability more than $\frac{1}{2}$, there will exist a player with a regret greater than $\frac{1}{2} - \frac{1}{\ell}$ in $p$. 
\end{theorem}  

\noindent
We can immediately apply Yao's minimax principle to this result to complete the proof.
\end{proof}


%% file: upload-section-coupled-arxiv.tex
\section{Achieving $\eps < \frac{1}{8}$ with Communication}
\label{sec:communicate}

We return to continuous dynamics to show that we can obtain a
worst-case regret of slightly less than $\frac{1}{8}$ by using limited
communication between players, thus breaking the uncoupled setting we
have been studying until now.

First of all, let us suppose that initially $p_i(0) = \frac{1}{2}$ for
each player $i$ and that \textbf{UCN} is run for $\frac{1}{2}$ time
units so that strategy/payoff states for each player lie on
$\mathcal{P} = \{ s_i \ | \ p_i^* = \frac{1}{2}(1 + D_i) \}$. We
recall from ~\Cref{lem:max-regret} that the worst case regret of
$\frac{1}{8}$ on this plane is achieved when $p_i^* = \frac{3}{4}$ and
$D_i = \frac{1}{2}$. We say a player is~\emph{bad} if they achieve a
regret of at least $0{\cdot}12$, which on $\mathcal{P}$ corresponds to having
$p_i^* \in [0{\cdot}7,0{\cdot}8]$. Similarly, all other players are
\emph{good}. We denote $\theta \in [0,1]$ as the proportion of players
that are bad. Furthermore, as the following lemma shows, we can in a
certain sense assume that $\theta \leq \frac{1}{2}$.

\begin{lemma}\label{lem:balancing}
  If $\theta > \frac{1}{2}$, then for a period of $0{\cdot}15$ time units, we
  can allow each bad player to shift to their best response with unit
  speed, and have all good players update according to \textbf{UCN} to
  stay on $\mathcal{P}$. After this movement, at most
  $1-\theta$ players are bad.
\end{lemma}

\begin{proof}
If $i$ is a bad player, in the worst case scenario, $\dot{D}_i = 2$,
which keeps their strategy/payoff state, $s_i$, on the plane $\mathcal{P}$.
However, at the end of $0{\cdot}15$ time units, they will have $p_i^* > 0{\cdot}85$,
hence they will no longer be bad. On the other hand, since the good players
follow the dynamic, they stay on $\mathcal{P}$, and at worst, all of them become bad.
\end{proof}

\begin{observation}
After this movement, players who were bad are the only players possibly away
from $\mathcal{P}$ and they have a discrepancy that is greater than
$0 \cdot 1$. Furthermore, all players who become bad lie on $\mathcal{P}$.
\end{observation}

We can now outline a continuous-time dynamic that utilises~\Cref{lem:balancing} to obtain a $(\frac{1}{8} - \frac{1}{220})$ maximal regret.

\begin{enumerate}
\item Have all players begin with $p_i(0) = \frac{1}{2}$
\item Run \textbf{UCN} for $\frac{1}{2}$ time units.
\item Measure, $\theta$, the proportion of bad players. If $\theta > \frac{1}{2}$ apply the dynamics of ~\Cref{lem:balancing}.
\item Let all bad players use $\dot{p}_i^* = 1$ for $\Delta = \frac{1}{220}$ time units.
\end{enumerate}

\begin{restatable}{theorem}{communication}
\label{thm:communication}
If all players follow the aforementioned dynamic, no single player will have a regret greater than $\frac{1}{8} - \frac{1}{220}$.
\end{restatable}

In essence one shows that if $\Delta$ is a small enough time interval (less than $0{\cdot}1$ to be exact), then all bad players will unilaterally decrease their regret by at least $0{\cdot}1\Delta$ and good players won't increase their regret by more than $\Delta$. The time step $\Delta = \frac{1}{220}$ is thus chosen optimally.

\begin{proof}
We have seen via ~\Cref{lem:balancing} that after step 3 the proportion of bad players is at most $\theta \leq \frac{1}{2}$, we wish to show that step 4 reduces maximal regret by at least $\frac{1}{220}$ for every bad player while maintaining a low regret for good players.

Since after step 3 all bad players remain on $\mathcal{P}$, we can consider an arbitrary bad player on the plane $\mathcal{P}$ with regret $r = D(1-p^*)$. Let us suppose that we allow all bad players to unilaterally shift their probabilities to their best response for a time period of $\Delta < 0{\cdot}4 \leq D$ units (the bound implies bad player preferences do not change). This means that  the worst case scenario for their regret is when their discrepancy increases to $D + 2\theta\Delta$. If we let $r'$ be their new regret after this move, we get the following:
$$
r' = (D + 2\theta\Delta)(1-p^*-\Delta) = D(1-p^*) + 2\theta\Delta(1-p^*) -D\Delta - 2\theta\Delta^2
$$
$$
= r -2\theta \Delta^2 + \left( 2\theta(1-p^*) - D \right)\Delta
$$
However, we can use our initial constraints on $D$ and $p^*$ from the fact that the players were bad, along with the fact that $\theta \leq \frac{1}{2}$ to obtain the following:
$$
2\theta(1-p^*) \leq (1-p^*) \leq 0{\cdot}3 < 0{\cdot}4 \leq D
$$
Hence as long as $\Delta < 0{\cdot}4$, $r' < r$ hence we can better the new bad players, without hurting the good players by choosing a suitably small value of $\Delta$.

To see that we don't hurt good players to much, suppose that we have a good player with discrepancy $D$ and best-response mass, $p^*$. By definition, their initial regret is $r = D(1-p^*) < 0{\cdot}12$. There are two extreme cases to what can happen to their regret after the bad players shift their strategies in step 4. Either their discrepancies increase by $2\theta \Delta$, in which case preferences are maintained, or either discrepancies decrease by $2\theta\Delta$ and preferences change (which can only occur when $2\theta\Delta > D$). For the first case we can calculate the new regret $r'$ as follows:
$$
r' = (D + 2\theta \Delta)(1-p^*) = r + 2\theta(1-p^*)\Delta \leq r + (1-p^*)\Delta \leq r+ \Delta
$$
This means that the total change in regret is at most $\Delta$. Note that if a player was originally bad and then shifted according to ~\Cref{lem:balancing} then their discrepancy is at least $0{\cdot}1$. For this reason if we limit ourselves to values of $\Delta < 0{\cdot}1$, then all such players will always fall in this case since their preferences cannot change.

Now we analyse the second case where preferences switch. Since we are only considering $\Delta < 0{\cdot}1$, then we can assume that all such profiles must lie on $\mathcal{P}$. In this case we get the following new regret:
$$
r' = (2\theta\Delta - D)(p) = r + 2\theta p^* \Delta - D \leq r + p^* \Delta  - D \leq r + \Delta
$$
Consequently, in the scenario that preferences change, the change of regret is bounded by $\Delta$ as well. This means that for $\Delta < 0{\cdot}1$, the decrease in regret for bad players is at least:
$$
2\theta \Delta^2 + (D - 2\theta(1-p^*))\Delta > 0{\cdot}1\Delta
$$
And for such time-steps $\Delta$, the regret for good players increases by at most $\Delta$. Thus under these bounds, the optimal value is $\Delta = \frac{1}{220}$ which gives rise to a maximal regret of $\frac{1}{8} - \frac{1}{220} = \frac{137}{1100}$.
\end{proof}

As a final note, we see that this process requires one round of
communication in being able to perform the operations
in~\Cref{lem:balancing}, that is we need to know if
$\theta > \frac{1}{2}$ or not to balance player profiles so that there
are at most the same number of bad players to good
players. Furthermore, in exactly the same fashion as
{\un($\alpha,\eta$)}, we can discretise the above process to
obtain a query-based algorithm that obtains a regret of
$\frac{1}{8} - \frac{1}{220} + \alpha < \frac{1}{8}$ for arbitrary
$\alpha$.


%% file: upload-section-extensions-arxiv.tex
\section{Extensions}\label{sec:extensions}

In this section we address two extensions to our previous results:
\begin{itemize}
\item (Section \ref{sec:various-lipschitz}) We extend the algorithm
  \textbf{UCN} to large games with a more general largeness parameter
  $\gamma = \frac{c}{n} \in [0,1]$, where $c$ is a constant.

\item (Section \ref{sec:block-update}) We consider large games with
$k$ actions and largeness parameter $\frac{c}{n}$ (previously we focused on $k=2$).
Our algorithm used a new uncoupled approach that is substantially different from the previous ones we have presented.
\end{itemize}

\subsection{Continuous Dynamics for Binary-action Games with  Arbitrary $\gamma$}\label{sec:various-lipschitz}

We recall that for large games, the largeness parameter $\gamma$
denotes the extent to which players can affect each others'
utilities. Instead of assuming that $\gamma = \frac{1}{n}$ we now let
$\gamma = \frac{c}{n} \in [0,1]$ for some constant $c$. We show that
we can extend \textbf{UCN} and still ensure a better than
$\frac{1}{2}$-equilibrium. We recall that for the original
\textbf{UCN}, players converge to a linear subspace
of strategy/payoff states and achieve a bounded regret. For arbitrary
$\gamma = \frac{c}{n}$, we can extend this subspace of strategy/payoff
states as follows:
\[
\mathcal{P}_\gamma = \Big\{ (p^*,D) \ | \ p^* = \min \left( \frac{1}{2} + \frac{D}{2c}, 1 \right) \Big\}
\]
where $D$ and $p^*$ represent respectively a player's discrepancy and
probability allocated to the best response. For $c = 1$ we recover the
subspace $\mathcal{P}$ as in \textbf{UCN}. Furthermore, if
$|\dot{{p}^*}| \leq 1$ for each player, then $| \dot{D} | \leq 2c$,
which means that we can implement an update as follows:
$$
\dot{{p}^*} = \frac{\dot{D}}{2c}
$$
This leads us to the following natural extension to~\Cref{thm:UNC}:
\begin{theorem}\label{thm:UNC-gamma}
Under the initial conditions $p_i(0) = \frac{1}{2}$ for all $i$, the following continuous dynamic, \textbf{UCN}-$\gamma$, has all players reach $\mathcal{P}_\gamma$ in at most $\frac{1}{2}$ time units. Furthermore, upon reaching $\mathcal{P}_\gamma$ a player never leaves.
\[
 \dot{{p}^*_i}(t) = f(D_i(t), \dot{D}_i(t)) = 
  \begin{cases} 
      \hfill \  1    \hfill & \text{ if } s_i \notin \mathcal{P}_\gamma  \\
      \hfill 0 \hfill & \text{ if } s_i \in \mathcal{P}_\gamma \text{ and } p^*_i >  \frac{1}{2} + \frac{D_i}{2c} \\
      \hfill \frac{\dot{D_i}}{2c} \hfill & \text{ otherwise }
  \end{cases}
\]
\end{theorem}

\noindent
Notice that unlike \textbf{UCN}, this dynamic is no longer necessarily a continuously differentiable function with respect to time when $c > 1$. However, it is still continuous. 

Once again, we note that for all strategy/payoff states, regret can be expressed as
$$
R = (1-p^*)D,
$$
from which  we can prove the following:

\begin{theorem}\label{thm:gamma-binary}
Suppose that $\gamma = \frac{c}{n}$ and that a player's strategy/payoff state lies on $\mathcal{P}_\gamma$, then her regret is at most $\frac{c}{8}$ for $c \leq 2$ and her regret is at most $\frac{1}{2} - \frac{1}{2c}$ for $c > 2$. Furthermore, the equilibria obtained are also $c$-WSNE.
\end{theorem} 

\begin{proof}
If $c \leq 2$, then regret is maximised when $D = \frac{c}{2}$ and consequently when $p^* = \frac{3}{4}$. This results in a regret of $\frac{c}{8}$. On the other hand, if $c > 2$, then regret is maximised when $D = 1$ and consequently $p^* = \frac{1}{2} + \frac{1}{2c}$. This results in a regret of $\frac{1}{2} - \frac{1}{2c}$.

As for the second part of the theorem, from the definition of $\mathcal{P}_\gamma$ and from the definition of $\epsilon$-WSNE in section \ref{def:WSNE} it is straightforward to see that when $D \geq c$, $p^* = 1$ which means that no weight is put on the strategy whose utility is at most $c$ from that of the best response.
\end{proof}

Thus we obtain a regret that is better than simply randomizing between both strategies, although as should be expected, the advantage goes to zero as the largeness parameter increases.

\subsubsection{Discretisation and Query Complexity}

In the same way as \textbf{UN}-$(\alpha,\eta)$, where we discretised \textbf{UN}, Theorem \ref{thm:gamma-binary} can be discretised to yield the following result.

\begin{theorem}
For a given accuracy parameter $\alpha$ and correctness probability $\eta$, we can implement a query-based discretisation of \textbf{UCN}-$\gamma$ that with probability $1-\eta$ correctly computes an $\epsilon$-approximate Nash equilibrium for

\[
 \epsilon = 
  \begin{cases} 
      \hfill \  \frac{c}{8} + \alpha    \hfill & \text{ if } c \leq 2  \\
      \hfill \frac{1}{2} - \frac{1}{2c} + \alpha \hfill & \text{ if } c > 2
  \end{cases}
\]
Furthermore the discretisation uses $O \left( \frac{1}{\alpha^4} \left( \frac{n}{\alpha \eta} \right) \right)$ queries. 
\end{theorem}

\subsection{Equilibrium Computation for $k$-action Games}\label{sec:block-update}

When the number of pure strategies per player is $k>2$, the initial ``strawman''
idea corresponding to Observation~\ref{obs:half} is to have all $n$ players randomize uniformly over their $k$ strategies.
Notice that the resulting regret may in general be as high as $1-\frac{1}{k}$.
In this section we give a new uncoupled-dynamics approach for computing approximate equilibria in $k$-action games where (for largeness parameter $\gamma=\frac{1}{n}$) the worst-case regret approaches $\frac{3}{4}$ as
$k$ increases, hence improving over uniform
randomisation over all strategies. Recall that in general we are considering
$\gamma = \frac{c}{n}$ for fixed $c \in [0,n]$.
The following is just a simple extension of the
payoff oracle $\cQ_{\beta,\delta}$ to the setting with $k$ actions:
for any input mixed strategy profile $p$, the oracle will with
probability at least $1 - \delta$, output payoff estimates for $p$
with error at most $\beta$ for all $n$ players.

\begin{paragraph}{Estimating payoffs for mixed profiles in $k$-action
    games.}
  Given a payoff oracle $\cQ$ and any target accuracy parameter
  $\beta$ and confidence parameter $\delta$, consider the following
  procedure to implement an oracle $\cQ_{\beta, \delta}$:
\begin{itemize}
\item For any input mixed strategy profile $p$, compute a new mixed
  strategy profile $p' = (1 - \frac{\beta}{2})p + (\frac{\beta}{2k})\mathbf{1}$ such
  that each player $i$ is playing uniform distribution with
  probability $\frac{\beta}{2}$ and playing distribution $p_i$ with
  probability $1 - \frac{\beta}{2}$.
\item Let $m = \frac{64k^2}{\beta^3} \log\left( 8n /\delta \right)$, and
  sample $m$ payoff queries randomly from $p'$, and call the oracle
  $\cQ$ with each query as input to obtain a payoff vector.
\item Let $\hat u_{i,j}$ be the average sampled payoff to player $i$
  for playing action $j$.\footnote{If the player $i$ never plays an
    action $j$ in any query, set $\hat u_{i,j} = 0$.} Output the
  payoff vector $(\hat{u}_{ij})_{i\in [n], j\in\{0, 1\}}$.
\end{itemize}
\end{paragraph}

As in previous sections, we begin by assuming that our algorithm has access to $\mathcal{Q}_M$, the more powerful query oracle that returns exact expected payoffs with regards to mixed strategies. We will eventually show in section \ref{sec:query-many} that this does not result in a loss of generality, as when utilising $\mathcal{Q}_{\beta, \delta}$ we incur a bounded additive loss with regards to the approximate equilibria we obtain.

The general idea of Algorithm~\ref{approx_nash_block} is as follows.
For a parameter $N\in\mathbb{N}$, every player uses a mixed strategy
consisting of a discretised distribution in which a player's probability
is divided into $N$ quanta of probability $\frac{1}{N}$, each of which
is allocated to a single pure strategy. We refer to these quanta as
``blocks'' and label them $B_1,\ldots,B_N$.
Initially, blocks may be allocated arbitrarily to pure strategies.
Then in time step $t$, for $t=1,\ldots,N$, block $t$ is reallocated
to the player's best response to the other players' current mixed
strategies.

The general idea of the analysis of Algorithm~\ref{approx_nash_block}
is the following. In each time step, a player's utilities change by
at most $n\gamma/N=c/N$. Hence, at the completion of
Algorithm~\ref{approx_nash_block}, block $N$ is allocated to a
nearly-optimal strategy, and generally, block $N-r$ is allocated to
a strategy whose closeness to optimality goes down as $r$ increases,
but enables us to derive the improved overall performance of each
player's mixed strategy.

\begin{algorithm} 
\caption{Block Equilibrium Computation \textbf{BU} (performed by each player)}
\label{approx_nash_block}                           
\begin{algorithmic}                    
    \REQUIRE 
    \STATE Parameter: $N \in \mathbb{N}$ 
    \STATE
	\INITIALIZATION
	\STATE Blocks $B_1,\ldots,B_N$; $\{$a block represents $\frac{1}{N}$ of the player's mixed strategy $\}$
	\STATE Allocate each $B_i$ to arbitrary $j\in[k]$;
	\STATE   
      \item[\textbf{Block Updates:}] 

        \FOR{$t = 1,\ldots,N$}

        \STATE Observe expected utilities
        $(u_{j})_{ j \in [k]}$ to the current mixed strategy profile
        $\vec{p}=(\vec{p}_i)_{i\in[n]}$;
        \STATE Reallocate $B_t$ to best response to $\vec{p}$;
                
    \ENDFOR   
    \STATE
    \RETURN $\vec{p}_i = (p_{j})_{ j \in [k]}$ 
  
\end{algorithmic}
\end{algorithm}

\begin{theorem}\label{first-bound}
\textbf{BU} returns a mixed strategy profile $(\vec{p}_i)_{i \in [n]}$ that is an $\epsilon$-NE when:
\[
 \epsilon =
  \begin{cases} 
      \hfill c \left( 1 + \frac{1}{N} \right)    \hfill & \text{ if } c \leq \frac{1}{2} \\
      \hfill 1 - \frac{1}{4c} + \frac{1}{2N} \hfill & \text{ if } c > \frac{1}{2} \\
  \end{cases}
\]
\end{theorem}

Notice for example that for $\gamma=\frac{1}{n}$ (i.e. putting $c=1$),
each player's regret is at most $\frac{3}{4}+\frac{1}{2N}$, so we
can make this arbitrarily close to $\frac{3}{4}$ since $N$ is a
parameter of the algorithm.

\begin{proof}
For an arbitrary player $i \in [n]$, in each step $t = 1,...,N$,
probability block $B_t$ is re-assigned to $i$'s current best response.

Since every player is doing the same transfer of probability, by the
largeness condition of the game, one can see that every block's assigned strategy incurs a regret that increases by at most $\frac{2c}{N}$ 
at every time step. This means that at the end of $N$ rounds, the
$j$-th block will at worst be assigned to a strategy that has
$\frac{\min \{1, (2c)j \}}{N}$ regret.
This means we can bound
a player's total regret as follows:
$$
R \leq \sum_{i=1}^N \frac{\min \{1, (2c)i \}}{N} \cdot \frac{1}{N}
$$  
There are two important cases for this sum: when $2c \leq 1$ and when $2c > 1$. In the first case:
$$
R \leq \sum_{i=1}^N \frac{2c i}{N^2} = n\gamma \left( 1 + \frac{1}{N} \right)
$$
And in the second:
$$
R \leq \left( \sum_{i=1}^{N/2c} \frac{2c i}{N^2} \right) + \left( N - \frac{N}{2c} \right) \cdot \frac{1}{N} = 1 - \frac{1}{4c} + \frac{1}{2N}
$$

\end{proof}

\noindent
In fact, we can slightly improve the bounds in Theorem \ref{first-bound}
via introducing a dependence on $k$.
In order to do so, we need to introduce some definitions first. 

\begin{definition} \label{truncated-triangles}
We denote $\mathcal{A}^{b,h}$ as the truncated triangle in the cartesian plane under the line $y = hx$ for $x \in [0,b]$ and height capped at $y = 1$. Note that if $bh \leq 1$ the truncated triangle is the entire triangle, unlike the case where $bh > 1$. See figure \ref{fig:triangle-def} for a visualisation.

\end{definition}

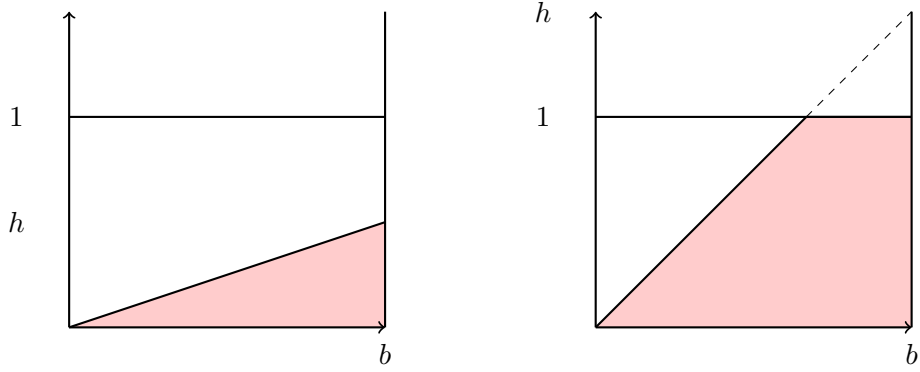
\begin{figure}
\center{
\begin{tikzpicture}[scale=0.7]
\tikzstyle{xxx}=[dashed,thick]

\fill[red!20](-8,2)--(-2,2)--(-2,4)--(-2,4)--cycle; 
\fill[red!20](2,2)--(8,2)--(8,6)--(6,6)--cycle; 

\draw[thick,->](-8,2)--(-2,2);
\draw[thick](-8,6)--(-2,6);

\draw[thick,->](2,2)--(8,2);
\draw[thick](2,6)--(8,6);

\draw[thick,->](-8,2)--(-8,8);
\draw[thick](-2,2)--(-2,8);

\draw[thick,->](2,2)--(2,8);
\draw[thick](8,2)--(8,8);

\draw[thick](-8,2)--(-2,4);

\draw[thick](2,2)--(6,6);
\draw[thick](6,6)--(8,6);
\draw[dashed](6,6)--(8,8);

\node at(-9,4){$h$};
\node at(1,8){$h$};
\node at(-9,6){$1$};
\node at(1,6){$1$};

\node at(8,1.5){$b$};
\node at(-2,1.5){$b$};

\end{tikzpicture}
\caption{Visualisation of $\mathcal{A}^{b,h}$ when $bh\leq 1$ (Left) and $bh > 1$ (Right).}\label{fig:triangle-def}
}
\end{figure}

\begin{definition} \label{left-sum}
For a given truncated triangle $\mathcal{A}^{b,h}$ and a partition of the base, $\mathcal{P} = \{x_1,...,x_r\}$ where $0\leq x_1 \leq \ldots \leq x_r \leq b$, we denote the left sum of $\mathcal{A}^{b,h}$ under $\mathcal{P}$ by $LS(\mathcal{A}^{b,h}, \mathcal{P})$ (for reference see figure \ref{fig:left-sum}) and define it as follows:
$$
LS(\mathcal{A}^{b,h}, \mathcal{P}) = \sum_{i=1}^{|\mathcal{P}|} (hx_i)(x_{i+1} - x_i)
$$
\end{definition}

\begin{figure}
\center{
\begin{tikzpicture}[scale=0.7]
\tikzstyle{xxx}=[dashed,thick]


\fill[red!20](0,2)--(2,2)--(2,3)--(0,3)--cycle; 
\fill[red!20](2,2)--(4.5,2)--(4.5,5)--(2,5)--cycle; 
\fill[red!20](4.5,2)--(8,2)--(8,7.5)--(4.5,7.5)--cycle; 
\fill[red!20](6,2)--(7,2)--(7,3.5)--(6,3.5)--cycle; 
\fill[red!20](7,2)--(8,2)--(8,8)--(7,8)--cycle; 
\fill[red!20](8,2)--(9,2)--(9,8)--(8,8)--cycle; 

\draw[thick, <->](-1,2)--(9,2);

\node at (-1.5,2){$0$};
\node at (-1.5,8){$1$};

\draw[thick, <->](-1,2)--(-1,8);

\draw[thick](9,2)--(9,8);

\draw[thick](-1,2)--(5,8);
\draw[thick](5,8)--(9,8);

\draw[thick](0,2)--(0,3);
\draw[thick](2,2)--(2,5);
\draw[thick](4.5,2)--(4.5,7.5);
\draw[thick](7,2)--(7,8);
\draw[thick](8,2)--(8,8);


\node at(-1,1.5){$0$};
\node at(0,1.5){$x_1$};
\node at(2,1.5){$x_2$};
\node at(5,1.5){$x_3$};
\node at(7,1.5){$x_4$};
\node at(8,1.5){$x_5$};
\node at(9,1.5){$b$};

\end{tikzpicture}
\caption{Example of left sum of five-element partition of base in the case where $bh>1$}\label{fig:left-sum}
}
\end{figure}
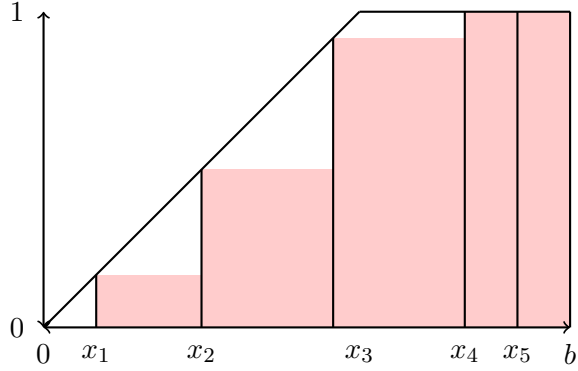

\noindent
With these definitions in hand, we can set up a correspondence between
the worst case regret of \textbf{BU} and left sums of
$\mathcal{A}^{(1+\frac{1}{N}), 2c}$. Suppose in the process of
\textbf{BU} a player has blocks $B_1,...,B_N$ in the
queue. Furthermore, without loss of generality, suppose that her $k$
strategies are sorted in ascending order of utility so that
$u_1,...,u_k$ where $u_j$ is the expected utility of the $j$-th
strategy at the end of the process. Furthermore, let $R_j = u_1 - u_j$
(i.e. the regret of strategy $j$), so that we also have
$0 = R_1 \leq R_2 \leq ... \leq R_k \leq 1$. If $N$ is much larger
than $k$, then by the pigeon-hole principle, many blocks will be
assigned to the same strategy, and hence will incur the same regret.
  However, as in the analysis of the
previous bounds, each block has restrictions as to how
much regret their assigned strategy can incur due to the largeness condition of the game. In
particular, the assigned strategy of block $B_b$ can only be
assigned to a strategy $j$ such that
$R_j \leq \min \{1, (2c)\} \cdot \left(\frac{b}{N} \right)$. For such
an assignment, since the block has probability mass $\frac{1}{N}$, it
contributes a value of
$R_j \cdot \left(\frac{j}{N} \right) \left(\frac{1}{N} \right)$ to the
overall regret of a player. Hence for fixed regret values
$(R_1,..,R_k)$, we can pick a valid assignment of these
values to blocks and get an expression for total regret that can
be visualised geometrically in figure \ref{fig:sup-optimal-k}.

The next important question is what valid assignment of blocks to
regret values results in the maximal amount of total regret for a
player. In figure \ref{fig:sup-optimal-k}, Block 1 is assigned to
strategy 1, Blocks 2,3, and 7 are assigned to strategy 2, blocks 4 and
5 are assigned to strategy 3, block 5 is assigned to strategy 4 and
finally blocks 8 and 9 are assigned to strategy 5.

One can see that this does not result in maximal regret. Rather it is simple to see that a greedy allotment of blocks to regret values results in maximal total regret. Such a greedy allotment can be described as follows: assign as many possible (their regret constraints permitting) blocks at the end of the queue to $R_k$, then repeat this process one-by-one for $R_i$ earlier in the queue. This is visualised in figure \ref{fig:optimal-k}, and naturally leads to the following result:

\begin{theorem}\label{reduction-triangle}
For any fixed $R_1,...,R_k$, the worst case assignment of probability blocks $B_b$ to strategies corresponds to a left sum of $\mathcal{A}^{(1+\frac{1}{N}), 2c}$ for some partition of $[0,1+\frac{1}{N}]$ with cardinality at most $k-1$.

\end{theorem}

\begin{figure}
\center{
\begin{tikzpicture}[scale=0.7]
\tikzstyle{xxx}=[dashed,thick]


\fill[red!20](0,2)--(1,2)--(1,2)--(0,2)--cycle; 
\fill[red!20](1,2)--(2,2)--(2,3.5)--(1,3.5)--cycle; 
\fill[red!20](2,2)--(3,2)--(3,3.5)--(2,3.5)--cycle; 
\fill[red!20](3,2)--(4,2)--(4,4.5)--(3,4.5)--cycle; 
\fill[red!20](4,2)--(5,2)--(5,6.5)--(4,6.5)--cycle; 
\fill[red!20](5,2)--(6,2)--(6,4.5)--(5,4.5)--cycle; 
\fill[red!20](6,2)--(7,2)--(7,3.5)--(6,3.5)--cycle; 
\fill[red!20](7,2)--(8,2)--(8,8)--(7,8)--cycle; 
\fill[red!20](8,2)--(9,2)--(9,8)--(8,8)--cycle; 

\draw[thick](0,2)--(8,2);

\draw[thick,<->](0,0.5)--(9,0.5);

\node at(0,-0.5){$0$};
\node at(9,-0.5){$1$};
\node at (-3,2){$0$};
\node at (-3,8){$1$};

\draw[thick](-1,2)--(-1,8);

\draw[thick,<->](-2,2)--(-2,8);

\draw[thick](9,2)--(9,8);

\draw[thick](-1,2)--(5,8);

\draw[thick](-1,8)--(9,8);
\draw[thick](-1,6.5)--(9,6.5);
\draw[thick](-1,4.5)--(9,4.5);
\draw[thick](-1,3.5)--(9,3.5);
\draw[thick](-1,2)--(9,2);

\draw[thick,->](0,2)--(0,3);
\draw[thick,->](1,2)--(1,4);
\draw[thick,->](2,2)--(2,5);
\draw[thick,->](3,2)--(3,6);
\draw[thick,->](4,2)--(4,7);
\draw[thick,->](5,2)--(5,8);
\draw[thick,->](6,2)--(6,8);
\draw[thick,->](7,2)--(7,8);
\draw[thick,->](8,2)--(8,8);

\node at(10,2){$R_1$};
\node at(10,3.5){$R_2$};
\node at(10,4.5){$R_3$};
\node at(10,6.5){$R_4$};
\node at(10,8){$R_5$};

\node at(0,1.5){$B_1$};
\node at(1,1.5){$B_2$};
\node at(2,1.5){$B_3$};
\node at(3,1.5){$B_4$};
\node at(4,1.5){$B_5$};
\node at(5,1.5){$B_6$};
\node at(6,1.5){$B_7$};
\node at(7,1.5){$B_8$};
\node at(8,1.5){$B_9$};

\draw (8,8) circle (2mm);
\draw (7,8) circle (2mm);
\draw (6,3.5) circle (2mm);
\draw (5,4.5) circle (2mm);
\draw (4,6.5) circle (2mm);
\draw (3,4.5) circle (2mm);
\draw (2,3.5) circle (2mm);
\draw (1,3.5) circle (2mm);
\draw (0,2) circle (2mm);

\end{tikzpicture}
\caption{For $N$ = 9 and $k$ = 5, and $c > \frac{1}{2}$, this shows a visualisation of a feasible allotments of regret values to blocks after \textbf{BU}. Note that this does not exhibit worst case regret.}\label{fig:sup-optimal-k}
}
\end{figure}
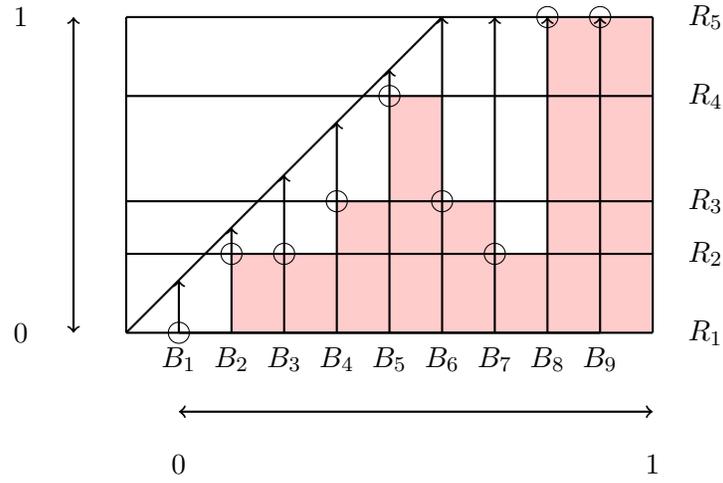

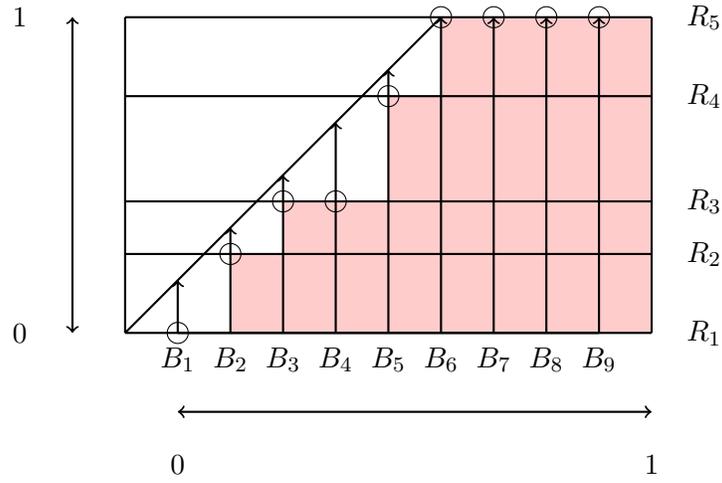
\begin{figure}
\center{
\begin{tikzpicture}[scale=0.7]
\tikzstyle{xxx}=[dashed,thick]


\fill[red!20](0,2)--(1,2)--(1,2)--(0,2)--cycle; 
\fill[red!20](1,2)--(2,2)--(2,3.5)--(1,3.5)--cycle; 
\fill[red!20](2,2)--(3,2)--(3,4.5)--(2,4.5)--cycle; 
\fill[red!20](3,2)--(4,2)--(4,4.5)--(3,4.5)--cycle; 
\fill[red!20](4,2)--(5,2)--(5,6.5)--(4,6.5)--cycle; 
\fill[red!20](5,2)--(6,2)--(6,8)--(5,8)--cycle; 
\fill[red!20](6,2)--(7,2)--(7,8)--(6,8)--cycle; 
\fill[red!20](7,2)--(8,2)--(8,8)--(7,8)--cycle; 
\fill[red!20](8,2)--(9,2)--(9,8)--(8,8)--cycle; 

\draw[thick](0,2)--(8,2);

\draw[thick,<->](0,0.5)--(9,0.5);

\node at(0,-0.5){$0$};
\node at(9,-0.5){$1$};
\node at (-3,2){$0$};
\node at (-3,8){$1$};

\draw[thick](-1,2)--(-1,8);

\draw[thick,<->](-2,2)--(-2,8);

\draw[thick](9,2)--(9,8);

\draw[thick](-1,2)--(5,8);

\draw[thick](-1,8)--(9,8);
\draw[thick](-1,6.5)--(9,6.5);
\draw[thick](-1,4.5)--(9,4.5);
\draw[thick](-1,3.5)--(9,3.5);
\draw[thick](-1,2)--(9,2);

\draw[thick,->](0,2)--(0,3);
\draw[thick,->](1,2)--(1,4);
\draw[thick,->](2,2)--(2,5);
\draw[thick,->](3,2)--(3,6);
\draw[thick,->](4,2)--(4,7);
\draw[thick,->](5,2)--(5,8);
\draw[thick,->](6,2)--(6,8);
\draw[thick,->](7,2)--(7,8);
\draw[thick,->](8,2)--(8,8);

\node at(10,2){$R_1$};
\node at(10,3.5){$R_2$};
\node at(10,4.5){$R_3$};
\node at(10,6.5){$R_4$};
\node at(10,8){$R_5$};

\node at(0,1.5){$B_1$};
\node at(1,1.5){$B_2$};
\node at(2,1.5){$B_3$};
\node at(3,1.5){$B_4$};
\node at(4,1.5){$B_5$};
\node at(5,1.5){$B_6$};
\node at(6,1.5){$B_7$};
\node at(7,1.5){$B_8$};
\node at(8,1.5){$B_9$};

\draw (8,8) circle (2mm);
\draw (7,8) circle (2mm);
\draw (6,8) circle (2mm);
\draw (5,8) circle (2mm);
\draw (4,6.5) circle (2mm);
\draw (3,4.5) circle (2mm);
\draw (2,4.5) circle (2mm);
\draw (1,3.5) circle (2mm);
\draw (0,2) circle (2mm);

\end{tikzpicture}
}
\caption{For $N$ = 9 and $k$ = 5, and $c > \frac{1}{2}$, this shows a visualisation of a feasible allotments of regret values to blocks after \textbf{BU}. Unlike figure \ref{fig:sup-optimal-k}, this does exhibit worst-case regret.}\label{fig:optimal-k}
\end{figure}

This previous theorem reduces the problem of computing worst case regret to that of computing maximal left sums under arbitrary partitions. To that end, we define the precise worst-case partition value we will be interested in.

\begin{definition}
For a given $\mathcal{A}^{b,h}$, let us denote the maximal left sum under partitions of cardinality $k$ by $\mathcal{A}^{b,h}_k$. Mathematically, the value is defined as follows:
$$
\mathcal{A}^{b,h}_k = \sup_{|\mathcal{P}| = k} LS(\mathcal{A}^{b,h}, \mathcal{P})
$$
\end{definition}

\noindent
We can explicity compute these values which in turn will bound a player's maximal regret.

\begin{lemma}
$\mathcal{A}^{1,1}_k = \left( \frac{1}{2}  \right) \left(\frac{k}{k+1} \right)$ which is obtained on the partition $\mathcal{P} = \{\frac{1}{k+1}, \frac{2}{k+1},...,\frac{k}{k+1}\}$
\end{lemma}

\begin{proof}

This result follows from induction and self-similarity of the original triangle. For $k = 1$, our partitions consist of a single point $x \in [0,1]$ hence the area under the triangle will be $\mathcal{A}^{1,1}_1(x) = (1-x)x$ which as a quadratic function of $x$ has a maximum at $x = \frac{1}{2}$. At this point we get $\mathcal{A}^{1,1}_1(x) = \frac{1}{2} \cdot \frac{1}{2}$ as desired.

Now let us assume that the lemma holds for $k = n$, we wish to show that it holds for $k = n+1$. Any $k = n+1$ element partition must have a left-most element, $x_1$. We let $\mathcal{A}'(x)$ be the maximal truncated area for an $n+1$ element partition, given that $x_1 = x$. By fixing $x$ we add an area of $x(1-x)$ under the triangle and we are left with $n$ points to partition $[x,1]$. We notice however that we are thus maximising truncated area under a similar triangle to the original that has been scaled by a factor of $(1-x)$. We can therefore use the inductive assumption and get the following expression:

$$
\mathcal{A}'(x) = (1-x)x + (1-x)^2 \mathcal{A}^{1,1}_n = (1-x)x + \frac{1}{2} (1-x)^2  \left( \frac{n}{n+1} \right)
$$

It is straightforward to see that $\mathcal{A}'(x)$ is maximised when $x = \frac{1}{k+2}$. Consequently the maximal truncated area arises from the partition where $x_i = \frac{i}{n+2}$ which in turn proves our claim.  

\end{proof}

\noindent
Via linear scaling, one can extend the above result to arbitrary base and height values $b,h$. 

\begin{corollary}\label{triangle-partition}
For $bh \leq 1$, $\mathcal{A}^{b,h}_k = \left( \frac{bh}{2} \right) \left( \frac{k}{k+1} \right)$ which is obtained on the partition $\mathcal{P} = \{\frac{b}{k+1}, \frac{2b}{k+1},...,\frac{kb}{k+1}\}$
\end{corollary}

\begin{corollary}\label{trapezoid-partition}
For $bh > 1$, we obtain the following expressions for  $\mathcal{A}^{b,h}_k$:

\[
 \mathcal{A}^{b,h}_k =
  \begin{cases} 
      \hfill \left( \frac{bh}{2} \right) \left( \frac{k}{k+1} \right)     \hfill & \text{ if }  \frac{k}{k+1} \leq \frac{b}{h} \\

      \hfill b( 1 - \frac{1}{h} - \frac{1}{2hk})     \hfill & \text{ otherwise } \\
  \end{cases}
\]

\end{corollary}

\begin{proof}

For the first case (when $\frac{k}{k+1} \leq \frac{b}{h}$), let us consider $\mathcal{B}^{b,h}$ to be the the triangle with base $b$ and height $h$ that unlike $\mathcal{A}^{b,h}$ is not truncated at unit height. From scaling our previous result from corollary \ref{triangle-partition}, the largest $k$-element left sum for $\mathcal{B}^{b,h}$ occurs for the partition $\mathcal{P} = \{\frac{b}{k+1}, \frac{2b}{k+1},...,\frac{bk}{k+1} \}$. However, from the fact that $\mathcal{A}^{b,h} \subset \mathcal{B}^{b,h}$, at precisely these values the left sums of $\mathcal{P}$ for both geometric figures coincide. It follows that this partition also gives a maximal $k$-element partition for left sums of $\mathcal{A}^{b,h}$ and thus the claim holds. 

On the other hand, let us know consider the case where $\frac{k}{k+1} > \frac{b}{h}$. In a similar spirit to previous proofs, let us define $\mathcal{A}(x): [0,b] \rightarrow \mathbb{R}$ to be the maximal left-sum  under $\mathcal{A}^{b,h}$ for a given partition $\mathcal{P}$ whose right-most element is $x$. From figures \ref{fig:sup-optimal-k} and \ref{fig:optimal-k}, it should be clear that we should only consider $x \in [0,\frac{b}{h}]$, because if ever we have a $x \geq \frac{b}{h}$, that would correspond to some block being assigned a regret value of $R_j = 1$ for some strategy $j$. However with the existence of such a maximal regret strategy, the greedy allotment of blocks to strategies would assign the most blocks possible to strategy $j$ (or some other maximal regret strategy), which would correspond again to the final element in our partition being $\frac{b}{h}$. 

Now that we have restricted our focus to $x \in [0,\frac{b}{h}]$, we wish to consider the triangle $\mathcal{B}^{\ell,\frac{k+1}{k}}$ of base length $\ell = \frac{(k+1)b}{kh}$, and height $\frac{k+1}{k}$ which is not truncated at height 1. Let us define $\mathcal{B}(x)$ to be a similar function that computes the maximal $k$-element left sum under $\mathcal{B}^{\ell,\frac{k+1}{k}}$ given that the right-most partition element is $x \in [0,\frac{•}{•}c{b}{h}]$. Geometrically, one can see that we get the following identity:

$$
\mathcal{A}(x) = \mathcal{B}(x) + \frac{hx}{b} \left( b - \frac{b}{h} \right)
$$
However, from corollary \ref{triangle-partition}, the optimal $k$-element partition on $\mathcal{B}^{\ell,\frac{k+1}{k}}$ has a right-most element of $\frac{\ell k}{k+1} = \frac{b}{h}$, it follows that $\mathcal{B}(x)$ is maximised at $x = \frac{b}{h}$. Furthermore, the second part of the above sum is also maximised at this value, therefore $\mathcal{A}(x)$ is maximised at $\frac{b}{h}$. Concretely, this means that the maximal $k$-element partition for $\mathcal{A}^{b,h}$ is $\mathcal{P} = \{\frac{b}{hk}, \frac{2b}{hk}, ...,\frac{(k-1)b}{hk}, \frac{b}{h}\}$. This partition results in a maximal left sum of $\mathcal{A}^{\frac{b}{h},1}_{k-1} + \left( b - \frac{b}{h} \right)$ which after simplification gives us the value $b( 1 - \frac{1}{h} - \frac{1}{2hk})$ as desired.

\end{proof}

\noindent
Finally, we can combine everything above to obtain:

\begin{theorem}\label{final-bounds-blocks}
With access to a query oracle that computes exact expected utilities for mixed strategy profiles, \textbf{BU} returns an $\epsilon$-approximate Nash equilibrium for

\[
 \epsilon =
  \begin{cases} 
      \hfill c \left( \frac{k-1}{k} \right) \left(1 + \frac{1}{N} \right)    \hfill & \text{ if } c \leq \frac{1}{2} \\
      \hfill c \left( \frac{k-1}{k} \right)\left(1 + \frac{1}{N} \right)     \hfill & \text{ if } c > \frac{1}{2} \text{ and } \frac{k-1}{k} \leq \frac{1}{2c} \\
      \hfill \left( 1 - \frac{1}{4c} - \frac{1}{4c (k-1)} \right) \left(1 + \frac{1}{N} \right) \hfill & \text{ if } c > \frac{1}{2}  \text{ and } \frac{k-1}{k} > \frac{1}{2c} \\
  \end{cases}
\]

\end{theorem}

\begin{proof}

This just a straightforward application of theorem \ref{reduction-triangle} and corollaries \ref{triangle-partition} and \ref{trapezoid-partition}.

\end{proof}

\subsubsection{Query Complexity of Block Method}\label{sec:query-many}

In the above analysis we assumed access to a mixed strategy oracle as we computed expected payoffs at each time-step for all players. When using $\cQ_{\beta, \delta}$ however, there is an additive error and a bounded correctness probability to take into account. 

In terms of the additive error, if we assume that there is an additive error of $\beta$ on each of the $N$ queries in \textbf{BU}, then at any time step, the $b$-th block will be assigned to a strategy that incurs at most $\left( \frac{\min \{1, (2c)b \}}{N} + \beta \right)$ regret, which can visualised geometrically in figure \ref{fig:additive-regret}, and which leads to the following extension of theorem \ref{reduction-triangle}.

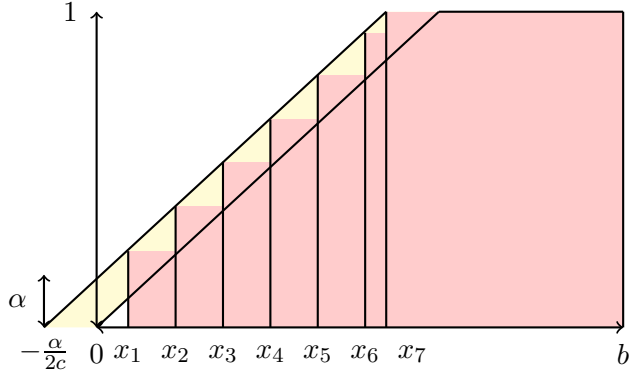
\begin{figure}
\center{
\begin{tikzpicture}[scale=0.7]
\tikzstyle{xxx}=[dashed,thick]


\fill[yellow!20](-2,2)--(-1,2)--(5.5,8)--(4.5,8)--cycle; 

\fill[red!20](-0.4,2)--(0.5,2)--(0.5,3.45)--(-0.4,3.45)--cycle; 
\fill[red!20](0.5,2)--(1.4,2)--(1.4,4.3)--(0.5,4.3)--cycle; 
\fill[red!20](1.4,2)--(2.3,2)--(2.3,5.15)--(1.4,5.15)--cycle; 
\fill[red!20](2.3,2)--(3.2,2)--(3.2,5.95)--(2.3,5.95)--cycle; 
\fill[red!20](3.2,2)--(4.1,2)--(4.1,6.8)--(3.2,6.8)--cycle; 
\fill[red!20](4.1,2)--(4.5,2)--(4.5,7.6)--(4.1,7.6)--cycle; 
\fill[red!20](4.5,2)--(8,2)--(8,8)--(4.5,8)--cycle; 
\fill[red!20](6,2)--(7,2)--(7,3.5)--(6,3.5)--cycle; 
\fill[red!20](7,2)--(8,2)--(8,8)--(7,8)--cycle; 
\fill[red!20](8,2)--(9,2)--(9,8)--(8,8)--cycle; 

\draw[thick, <->](-1,2)--(9,2);

\node at (-1.5,8){$1$};

\draw[thick, <->](-1,2)--(-1,8);

\draw[thick, <->](-2,2)--(-2,3);
\node at (-2.5,2.5){$\alpha$};

\draw[thick](9,2)--(9,8);

\draw[thick](-1,2)--(5.5,8);
\draw[thick](5.5,8)--(9,8);

\draw[thick](-2,2)--(4.5,8);

\draw[thick](-0.4,2)--(-0.4,3.45);
\draw[thick](0.5,2)--(0.5,4.3);
\draw[thick](1.4,2)--(1.4,5.15);
\draw[thick](2.3,2)--(2.3,5.95);
\draw[thick](3.2,2)--(3.2,6.8);
\draw[thick](4.1,2)--(4.1,7.6);
\draw[thick](4.5,2)--(4.5,8);


\node at(-2,1.5){$-\frac{\alpha}{2c}$};
\node at(-1,1.5){$0$};
\node at(-0.4,1.5){$x_1$};
\node at(0.5,1.5){$x_2$};
\node at(1.4,1.5){$x_3$};
\node at(2.3,1.5){$x_4$};
\node at(3.2,1.5){$x_5$};
\node at(4.1,1.5){$x_6$};
\node at(5,1.5){$x_7$};
\node at(9,1.5){$b$};

\end{tikzpicture}
\caption{Example of $\alpha$ additive error in utility sampling. For this 7 element partition, regret bounds are increased by $\alpha$ and we get an augmented truncated triangle.}\label{fig:additive-regret}
}
\end{figure}

\begin{theorem}\label{reduction-triangle-additive-error}
In \text{BU}, if queries incorporate an additive error of $\alpha$ on expected utilities, for any fixed choice of $R_1,...,R_k$, the worst case assignment of probability blocks $B_b$ to strategies corresponds to a left sum of $\mathcal{A}^{(1+\frac{1}{N} + \frac{\beta}{2c}), 2c}$ for some partition of $[0,1+\frac{1}{N}]$ with cardinality at most $k-1$.

\end{theorem}

Finally, since our approximate query oracle is correct with a bounded probability, in order to assure that the same additive error of $\beta$ holds on all $N$ queries of \text{BU}, we need to impose a correctness probability of $\frac{\delta}{N}$ in order to achieve the former with a union bound. This leads to the following query complexity result for \textbf{BU}.

\begin{theorem}\label{query-BU}
For any $\alpha, \eta > 0$, if we implement \textbf{BU} using $\cQ_{\beta,\delta}$ with $\beta = \alpha$ and $\delta = \frac{\eta}{N}$, with probability $1 - \eta$, we will obtain an $\epsilon$-approximate Nash equilibrium for

\[
 \epsilon =
  \begin{cases} 
      \hfill c \left( \frac{k-1}{k} \right) \left(1 + \frac{1}{N}  + \frac{\alpha}{2c} \right)    \hfill & \text{ if } c \leq \frac{1}{2} \\
      \hfill c \left( \frac{k-1}{k} \right)\left(1 + \frac{1}{N} + \frac{\alpha}{2c}\right)     \hfill & \text{ if } c > \frac{1}{2} \text{ and } \frac{k-1}{k} \leq \frac{1}{2c} \\
      \hfill \left( 1 - \frac{1}{4c} - \frac{1}{4c (k-1)} \right) \left(1 + \frac{1}{N} + \frac{\alpha}{2c} \right) \hfill & \text{ if } c > \frac{1}{2}  \text{ and } \frac{k-1}{k} > \frac{1}{2c} \\
  \end{cases}
\]
The total number of queries used is $\frac{64k^2}{\alpha^3} \log\left( \frac{8nN}{\delta} \right)$

\end{theorem}

Once again, it is interesting to note that the first regret bounds we derived do not depend on $k$. It is also important to note the regret has an extra term of the form $O(\frac{1}{N})$ in the number of probability blocks. Although this can be minimised in the limit, there is a price to be paid in query complexity, as this would involve a larger number of rounds in the computation of approximate equilibria.

\subsection{Comparison Between Both Methods}

We can compare the guarantees from our methods from \ref{sec:various-lipschitz} and \ref{sec:block-update} when we let the number of strategies $k = 2$ and we consider largeness parameters $\gamma = \frac{c}{n} \in [0,1]$. Furthermore, we consider how both methods compare when $N\rightarrow \infty$.

\begin{center}

\begin{tabular}{|c|c|c|c|}
 \hline 
 • & $c \leq 1$ & $1 \leq c \leq 2$ & $c \geq 2$ \\ 
 \hline 
 \textbf{UNC} & $\frac{c}{8}$ & $\frac{c}{8}$ & $\frac{1}{2} - \frac{1}{2c}$ \\ 
 \hline 
 \textbf{BU} & $\frac{c}{2}$ & $1 - \frac{1}{2c}$ & $1 - \frac{1}{2c}$ \\ 
 \hline 
 \end{tabular}  

 \end{center}

One can see that \textbf{UNC} does better by a multiplicative factor of $\frac{1}{4}$ in the case of small $c$ and better by an additive factor of $\frac{1}{2}$ for large $c$.


%% file: upload-section-conclusion-arxiv.tex
\section{Conclusion and Further Research}

The obvious question raised by our results is the possible improvement
in the additive approximation obtainable.
Since {\em pure} approximate equilibria are known to exist for these games,
the search for such equilibria is of interest.
A slightly weaker objective (but still stronger than the solutions we obtain here)
is the search for {\em well-supported} approximate equilibria in cases where $c > 1$ and for better {\em well-supported} approximate equilibria in general.

There is also the question of lower bounds, especially in the completely
uncoupled setting.
Our algorithms are randomised (estimating the payoffs that result from
a mixed strategy profile via random sampling) and one might also ask
what can be achieved using deterministic algorithms.